\documentclass[11pt]{article}

\usepackage{fullpage}
\usepackage{amsmath}
\usepackage{amsfonts}
\usepackage{amssymb}
\usepackage{amsthm}
\usepackage{times}
\usepackage{latexsym}
\usepackage{color}
\usepackage{enumitem} 
\usepackage[colorlinks=true]{hyperref}
\usepackage{tabu}
\usepackage[final]{graphicx}
\usepackage{tikz}
\usepackage{color,soul}

\usepackage{epsfig}

\usepackage{thmtools}
\usepackage{thm-restate}

\usepackage{hyperref}

\usepackage{bbm}
\usepackage{dsfont}

\usepackage{graphicx}
\usepackage{algorithm}
\usepackage[noend]{algpseudocode}

\usetikzlibrary{trees}

\newtheorem{theorem}{Theorem}[section]
\newtheorem{lemma}[theorem]{Lemma}
\newtheorem{corollary}[theorem]{Corollary}

\theoremstyle{definition}
\newtheorem{definition}[theorem]{Definition}

\theoremstyle{remark}
\newtheorem{remark}[theorem]{Remark}

\makeatletter
\def\BState{\State\hskip-\ALG@thistlm}
\makeatother

\numberwithin{equation}{section}

\newcommand{\R}{\mathbb{R}}

\newcommand{\states}{S}
\newcommand{\actions}{A}
\newcommand{\valuespace}{\R^{\states}}
\newcommand{\policyspace}{\actions^\states}
\newcommand{\offsetspace}{\R^{\states \times \actions}}
\newcommand{\defeq}{\stackrel{\mathrm{\scriptscriptstyle def}}{=}}
\newcommand{\valop}{T}

\newcommand{\Prob}{\mathbb{P}}
\newcommand{\E}{\mathbb{E}}

\newcommand{\apxvalcode}{\texttt{ApxVal}}
\newcommand{\apxtranscode}{\texttt{ApxTrans}}

\newcommand{\randomVI}{\texttt{RandomizedVI}}
\newcommand{\sampledRandomVI}{\texttt{SampledRandomizedVI}}
\newcommand{\sampledRandomMonVI}{\texttt{SampledRandomizedMonVI}}
\newcommand{\highprecisionVI}{\texttt{HighPrecisionRandomVI}}
\newcommand{\sublinearRandomVI}{\texttt{SublinearRandomVI}}
\newcommand{\sublinearRandomMonVI}{\texttt{SublinearRandomMonVI}}

\newcommand{\finiteHVI}{\texttt{RandomizedFiniteHorizonVI}}
\newcommand{\varianceReducedFiniteHVI}{\texttt{VarianceReducedFiniteHorizonVI}}

\newcommand{\apxmonvalcode}{\texttt{ApxMonVal}}
\newcommand{\otilde}{\tilde{O}}
\newcommand{\vopt}{v^{*}}

\newcommand{\vones}{\vec{1}}

\newcommand{\failprob}{\delta}

\newcommand{\runtimenearline}{ \otilde \left( \left(|S|^2 |A|  + \frac{|S| |A|}{(1 - \gamma)^3} \right)
 \log\left( \frac{M}{\epsilon} \right) \log\left( \frac{1}{\delta} \right) \right)}

\newcommand{\runtimesublin}{\otilde \left(\frac{|S| |A| M^2}{(1 - \gamma)^4 \epsilon^2} \log \left(\frac{1}{\delta}\right) \right)}

\newcommand{\runtimeip}{\otilde(|S|^{2.5} |A|  \log(M/((1-\gamma)\epsilon)))}
\newcommand{\runtimeipnice}{\otilde\left(|S|^{2.5} |A| \log \left( \frac{M}{(1 - \gamma)\epsilon}\right)\right)}

\newcommand{\innerNumIter}{L}
\newcommand{\innerIter}{l}

\DeclareMathOperator*{\argmax}{argmax}

\def\S{|S|}
\def\A{|A|}
\newcommand{\cO}{O}
\newcommand{\tO}{\otilde}	

\begin{document}

\title{Variance Reduced Value Iteration and\\ Faster Algorithms for Solving Markov Decision Processes}

\author{
	Aaron Sidford \\
	Stanford University \\ 
	\texttt{sidford@stanford.edu}
	\and
	Mengdi Wang \\
	Princeton University \\
	\texttt{mengdiw@princeton.edu}
	\and
	Xian Wu \\
	Stanford University \\
	\texttt{xwu20@stanford.edu}
	\and
	Yinyu Ye \\
	Stanford University \\
	\texttt{yyye@stanford.edu}
	}

\date{}

\maketitle

\begin{abstract}
In this paper we provide faster algorithms for approximately solving discounted Markov Decision Processes in multiple parameter regimes. Given a discounted Markov Decision Process (DMDP) with $|\states|$ states, $|\actions|$ actions, discount factor $\gamma\in(0,1)$, and rewards in the range $[-M, M]$,
we show how to compute an $\epsilon$-optimal policy, with probability $1 - \delta$ in time\footnote{We use $\otilde$ to hide polylogarithmic factors in the input parameters, i.e. $\otilde(f(x)) = O(f(x) \cdot \log(f(x))^{O(1)})$.}
\[
\runtimenearline
~ .
\]
This contribution reflects the first nearly linear time, nearly linearly convergent algorithm for solving DMDPs for intermediate values of $\gamma$.  

We also show how to obtain improved sublinear time algorithms provided we can sample from the transition function in $O(1)$ time. Under this assumption we provide an algorithm which computes an $\epsilon$-optimal policy for $\epsilon \in (0, \frac{M}{\sqrt{1-\gamma}}]$ with probability $1 - \delta$ in time 
 \[
\runtimesublin ~.
 \]
 Further, we extend both these algorithms to solve finite horizon MDPs. Our algorithms improve upon the previous best for approximately computing optimal policies for fixed-horizon MDPs in multiple parameter regimes.

Interestingly, we obtain our results by a careful modification of approximate value iteration. We show how to combine classic approximate value iteration analysis with new techniques in variance reduction. Our fastest algorithms leverage further insights to ensure that our algorithms make monotonic progress towards the optimal value. This paper is one of few instances in using sampling to obtain a linearly convergent linear programming algorithm and we hope that the analysis may be useful more broadly. 
\end{abstract} 

\newpage

\section{Introduction}

Markov decision processes (MDPs) are a popular mathematical framework used to encapsulate sequential decision-making problems under uncertainty. MDPs are widely used to formulate problems that require planning to maximize rewards aggregated over a long time horizon. Rooted from classical physics, MDPs are the discrete analog of variational problems and continuous-state stochastic optimal control problems and are a fundamental computational model used in the study of control, dynamical systems, and artificial intelligence, with applications in many industries, including health care, finance and engineering. Moreover, recently MDPs have become increasingly important in reinforcement learning, a rapidly developing area of artificial intelligence that studies how an agent interacting with a poorly understood environment can learn over time to make optimal decisions. 

Although MDPs have been extensively studied across multiple disciplines since the 1950s, the development of efficient algorithms for large-scale MDPs remains challenging. More precisely, the best known algorithms in many parameter regimes scale super-linearly with the size of the input which can be prohibitively difficult, especially in large-scale applications. In this work, we provide the first nearly linear convergent, nearly linear time algorithms that can solve MDPs with high precision even when the discount factor depends polynomially (with a small exponent) on the number of states and actions. We also give sublinear algorithms with the fastest runtime dependencies on the parameters of the discounted MDP -- the state and action space, and the discount factor. Our algorithms combine the classic value iteration algorithm with novel sampling and variance reduction techniques and our results provide new complexity benchmarks for solving discounted infinite-horizon MDPs.

In this work, we focus on the discounted infinite-horizon Markov decision problem (DMDP). DMDPs are described by the tuple $(S, A, P, r, \gamma)$, where $S$ is the finite state space, $A$ is the finite action space, and $\gamma \in(0,1)$ is the discount factor. $P$ is the collection of state-action-state transition probabilities, with each $p_{a}(i,j)$ specifying the probability of going to state $j$ from state $i$ when taking action $a$, $r$ is the collection of rewards at different state-action pairs, i.e., we collect $r_{a}(i)$ if we are currently in state $i$ and take action $a$. 

In a Markov Decision Process, at each time step $t$, a controller takes an available action $a\in A$ from their current state $i$ and reaches the next state $j$ with probability $p_{a}(i,j)$. For each action $a$ taken, the decision maker earns an immediate reward, $r_{a}(i)$. A vector $\pi\in \policyspace$ that tells the actor which action to take from any state is called a stationary policy. 
 
The main goal in solving a DMDP is to find 
a stationary, deterministic 
policy $\pi^*$ that maximizes the expected discounted cumulative reward. A stationary policy specifies
actions to follow at each state, irregardless of time. A deterministic policy gives a single fixed prescribed action for each state. For all results and analyses regarding DMDPs
 in the rest of this paper, all our policies will be stationary and deterministic. The goal of maximizing the expected discounted cumulative reward can be formulated as
$$ \max_{\pi \in A^S} v_\pi(i) := \E_{\pi}\left[ \sum^{\infty}_{t=1} \gamma^t r_{a_{t} }(i_t) \mid i_0 = i\right],$$
where $\{i_0, a_0, i_1,a_1, \ldots,i_t,a_t,\ldots\}$ are state-action transitions generated by the MDP under the fixed policy $\pi$, i.e. $a_t = \pi_{i_t}$,
 and the expectation $ \E_{\pi}[\cdot]$ is over the set of $(i_t,a_t)$ trajectories. 

We refer to $v_{\pi}$ as the expected value vector of policy $\pi$, which describes the expected discounted cumulative rewards corresponding to all possible initial states. We refer to $\vopt$ value vector associated with $\pi^*$, the policy that maximizes expected discounted cumulative reward for each initial state.
We are interested in finding an $\epsilon$-optimal policy $\pi$, under which the expected cumulative reward is $\epsilon$-close to the maximal expected cumulative reward regardless of the initial state. More precisely, we say  $\pi$ is $\epsilon$-optimal if $\| \vopt - v_{\pi}\|_{\infty} \leq \epsilon$. 

\subsection{Our Results}

In this paper we provide several randomized algorithms for approximately solving DMDPs, i.e. computing approximately optimal policies. To achieve our fastest running times we combine classic sampling techniques for
 value iteration with
new methods associated with variance reduction to obtain faster running times. Our fastest algorithms leverage further insights to ensure that our algorithms make monotonic progress towards the optimal value. (See Section~\ref{approach_overview} for a more detailed overview of our approach.)

Our main results are an algorithm that computes an $\epsilon$-approximate policy with probability $1 - \delta$ in time\footnote{We use $\otilde$ to hide polylogarithmic factors in the input parameters, i.e. $\otilde(f(x)) = O(f(x) \cdot \log(f(x))^{O(1)})$.} 
\[
\runtimenearline ~,
\]
and another algorithm which given a data structure from which we can sample state-action-state tuple according to its transition probability in expected $O(1)$ time,\footnote{The actual time it takes to sample by transition probabilities depends on the precise arithmetic model, assumptions about how the input is given, and what preprocessing is performed. For example with $\otilde(|S|^2 |A|)$ preprocessing, transitions can be stored in binary trees that allow sampling to be performed in $O(\log(|S|))$ time. However by using arrays and assuming indexing takes $O(1)$ time, faster sampling times can be achieved. We assume the sampling time is $O(1)$ to simplify notation. If sampling instead takes time $O(\alpha)$ this increases our sublinear running times by at most a multiplicative $O(\alpha)$.} 
computes an $\epsilon$-approximate policy for $\epsilon \in (0, \frac{M}{\sqrt{1-\gamma}}]$  in time
\[
\runtimesublin ~.
\]
The first algorithm is nearly linearly convergent, i.e., has an $\otilde(\log(1/\epsilon))$ dependence on $\epsilon$, and runs in nearly-linear time whenever $1/(1 - \gamma) = O(|S|^{1 / 3})$.
 This is the first algorithm for solving DMDPs that provably runs in nearly linear time even when the discount factor can be polynomial in the number of states and actions (albeit with a small exponent). Notably, it either matches or improves upon the performance of various forms of value iteration \cite{tseng1990solving, littman1995complexity} in terms of the dependence on $|S|$, $|A|$, and $\gamma$.\footnote{This algorithm also uses an oracle for sampling transition probabilities in expected $O(1)$ time. Since this can be implemented with $\otilde(|S|^2|A|)$ preprocessing time, it does not affect the algorithm's asymptotic runtime and no sampling assumption is needed.}

The second algorithm is sublinear in the $\Omega(|S|^2 |A|)$ input size of the DMDP
 as long as $\frac{1}{\epsilon^2} = o(\frac{|S|(1-\gamma)^4}{M^2})$.
 It improves upon existing results on sampling-based value iteration (Q-learning) in its dependence on $1/(1-\gamma)$ \cite{kearns1999finite}. The algorithm matches the running time of \cite{wang2017randomized} under fewer assumptions (i.e. we make no assumptions regarding ergodicity). 

We also extend these two algorithms to solve finite horizon MDPs. To the best of our knowledge, our algorithms are the first application of variance reduction for approximate value iteration to solve for approximately optimal policies for finite horizon MDPs. In Table~\ref{tab:literature_runtime_exact} and Table~\ref{tab:literature_runtime_apx} we provide a more comprehensive comparison of running time and in Section~\ref{sec:prev_work} we provide a more comprehensive discussion of previous work.

\subsection{Approach Overview}\label{approach_overview}

We achieve our results by building on the classic \emph{value iteration} algorithm. This algorithm simply computes a sequence of values $v_k \in \valuespace$, $k=0,1,\ldots,$ by applying the rule
\[
[v_{k + 1}]_i =  \max_{a \in \actions} \left[ r_{a}(i) + \gamma \cdot p_a(i)^\top v_{k} \right]
\]
for all $i \in \states$ and $k \geq 0$, and using the maximizing action as the current policy, where we denote by $p_a(i)\in\valuespace$ the vector of transition probabilities given by $p_a(i) = [p_a(i,j)]_{j\in S}$. In our adaption, instead of exactly computing $p_a(i)^\top v_{k}$, we approximate it by sampling. We analyze the performance using Hoeffding's inequality and well-known contraction properties of value iteration. While this is a fairly classic idea, we show how to improve it by using a type of \emph{variance reduction}.

More formally, instead of sampling to compute $p_a(i)^\top v$ for current values $v \in \valuespace$ and all actions $a \in \actions$ and $i \in \states$ from scratch in every iteration, we show that we can decrease the number of samples needed by instead estimating how these quantities change over time. Our faster variance reduced algorithms first compute fairly precise estimates for $p_a(i)^\top v_0$ for initial values $v_0 \in \valuespace$ and all actions $a \in \actions$ and $i \in \states$ and then, when after $k$ iterations the current value vector is $v_k$, sampling is used to estimate $p_a(i)^\top (v_k - v_0)$ for all actions $a \in \actions$ and $i \in \states$. Adding these estimated changes to the previously estimated values of $p_a(i)^\top v$ is our approximate value iterate. 

By sampling the change in $p_a(i)^\top v_k$ from some fixed $p_a(i)^\top v_0$, fewer samples are needed, compared to the sampling complexity of simply estimating $p_a(i)^\top v_k$, due to lower variance. We show that as with other applications of variance reduction, this scheme ties the incurred error to the current quality of the $v_k$. Exploiting this fact and carefully trading off how often the $p_a(i)^\top v_0$ are estimated and how the differences $p_a(i)^\top (v_k - v_0)$ are estimated, is what leads to our reductions in running time. 
Variance reduction has recently become a popular technique for obtaining faster algorithms for a wide range of optimization problems, and our application is the first for Markov Decision Processes \cite{Johnson013}.

To further accelerate the sampling algorithm for computing $\epsilon$-optimal policies, we modify our algorithm to start with an underestimate of the true values and increase them monotonically towards the optimum. We show that this allows us to tighten our runtime bounds and analyses in two ways. First, the fact that values are always increasing directly reduces the number of samples required. Secondly, it helps us maintain the invariant that the current proposed stationary policy $\pi$ and values $v$ satisfy the inequality $v \leq v_{\pi}$ entry-wise. This ensures that the induced policy has true values greater then $v$ and thus allows us to convert from approximate values to approximate policies without a loss in approximation quality. 

We hope these techniques will be useful in the development of even faster MDP algorithms in theory and practice. We are unaware of variance reduction as employed by this algorithm having been used previously to obtain linearly convergent algorithms for linear programming and thus we hope that the analysis of this paper may open the door for faster algorithms for a broader set of convex optimization problems.

\section{Previous Work}
\label{sec:prev_work}

The design and analysis of algorithms for solving MDPs have interested researchers across multiple disciplines since the 1950s. There are three major deterministic approaches: value iteration, policy iteration method, and linear programming, for finding the exact optimal stationary policy of an MDP \cite{bertsekas1995dynamic}.  Later in 1990s, sampling-based methods for MDP began to gain traction, and lay the foundation for reinforcement learning algorithms today. Despite years of study, the complexity of MDP remains an open question and the best known running times are far from linear. 

\paragraph{Deterministic Methods for MDP}
Bellman \cite{bellman1957dynamic} developed value iteration as a successive approximation method to solve nonlinear fixed-point equations. Its convergence and complexity have been thoroughly analyzed; see e.g. \cite{tseng1990solving, littman1995complexity}.  It is known that value iteration can compute an exact solution to a DMDP in time $\cO(\S^2\A L \frac{\log(1/(1-\gamma))}{1-\gamma})$, where $L$ is a measure of complexity of the associated linear program that is at most the number of bits needed to represent the input. It is also well known that value iteration can find an approximate $\epsilon$-approximate solution in time $\cO(\S^2\A \frac{\log(1/\epsilon(1-\gamma))}{1-\gamma})$.

Howard introduced policy iteration shortly thereafter \cite{howard1960dynamic}, and its complexity has also been analyzed extensively; see e.g. \cite{mansour1999complexity,ye2011simplex,scherrer2013improved}. Not long after the development of value iteration and policy iteration, \cite{d1963probabilistic} and \cite{de1960problemes} discovered that the Bellman equation can be formulated into an equivalent linear program, allowing a rich suite of tools developed for LP, such as interior point methods and the simplex method by Dantzig \cite{dantzig2016linear}, to solve MDPs exactly. This connection also led to the insight that the simplex method, when applied to solving DMDPs, is the simple policy iteration method. 

\cite{feinberg2014value} showed that value iteration is not strongly polynomial for DMDP. However, Ye \cite{ye2011simplex} showed that policy iteration (which is a variant of the general simplex method for linear programming) and the simplex method are strongly polynomial for DMDP and terminates in $\cO(\frac{|S|^2|A|}{1-\gamma}\log(\frac{\S}{1-\gamma}))$ number of iterations. 
\cite{hansen13} and \cite{scherrer2013improved} improved the iteration bound to $O(\frac{|S||A|}{(1-\gamma)}\log(\frac{1}{(1-\gamma)}))$ for Howard's policy iteration method.
\cite{ye2005new} also designed a combinatorial interior-point algorithm (CIPA) that solves the DMDP in strongly polynomial time. 

Recent developments \cite{lee2014path, lee2015efficient}
showed that linear programs can be solved in $\tO(\sqrt{\hbox{rank}(A)} )$ number of linear system solves, which, applied to DMDP, leads to a running time of $\tO( \S^{2.5}\A L )$ and $\runtimeip$ (see Appendix~\ref{sec:dmdp_with_ip} for a derivation).

We also note that while there are many methods for approximate linear programming, the error they incur for solving DMDPs is unclear as care needs to be taken in converting between $\epsilon$-approximate linear programming solutions and $\epsilon$-approximate values and policies (note that $\epsilon$-approximate values does not necessarily lead to $\epsilon$-approximate policies). For an illustrative example, see Appendix~\ref{sec:dmdp_with_ip} where we show how convert a particular type of approximate linear programming solution to an approximate policy and use interior point methods to compute approximate policies efficiently.

For more detailed surveys on MDP and its solution methods, we refer the readers to the textbooks \cite{bertsekas1995dynamic, bertsekas1995neuro, puterman2014markov, bertsekas2013abstract} and the references therein. In particular, our treatment of the Bellman operator and its properties follows the style set by \cite{bertsekas1995dynamic, bertsekas2013abstract} for deterministic value iteration.

\begin{table}[h!]
\begin{center}
	\begin{tabular}{|c|c|c|}
		\hline
		Value Iteration & $\S^2\A L \frac{\log(1/(1-\gamma))}{1-\gamma}$ & \cite{tseng1990solving, littman1995complexity}
		\\ \hline
		Policy Iteration (Block Simplex) &  $\frac{\S^4\A^2}{1-\gamma} \log(\frac1{1-\gamma})$ & \cite{ye2011simplex},\cite{scherrer2013improved}
		\\  \hline
		Recent Interior Point Methods  &   $\tO( \S^{2.5}\A L )$ &\cite{lee2014path}
		\\  \hline
		Combinatorial Interior Point Algorithm   &    $\S^4\A^4\log\frac{\S}{1-\gamma}$ & \cite{ye2005new}
		\\  \hline
	\end{tabular}
\end{center}
\caption{\textbf{Running Times to Solve DMDPs Exactly}: In this table, $\S$ is the number of states, $\A$ is the number of actions per state, $\gamma \in(0,1)$ is the discount factor, and $L$ is a complexity measure of the linear program formulation that is an upper bound on the total bit size to present the DMDP input.}
\label{tab:literature_runtime_exact}
\end{table}

\begin{table}[h!]
\begin{center}
	\begin{tabular}{|c|c|c|}
		\hline
		Value Iteration & $\S^2\A \frac{\log(1/(1-\gamma)\epsilon)}{1-\gamma}$ & \cite{tseng1990solving, littman1995complexity}
		\\ \hline
		Recent Interior Point Methods  &   $\runtimeipnice$ &\cite{lee2014path} (Appendix~\ref{sec:dmdp_with_ip})
		\\  \hline
		{Randomized Primal-Dual Method}    & $\tO (C\frac{\S \A M^2}{(1-\gamma)^4\epsilon^2}) $  & \cite{wang2017randomized}
		\\  \hline
		{Empirical QVI}    & $\tO \left(\frac{\S \A M^2}{(1-\gamma)^3\epsilon^2}\right) $  if $\epsilon \in (0, \frac{M}{\sqrt{(1-\gamma)|S|}}]$& \cite{azar2013minimax}
		\\  \hline
		High Precision Randomized Value Iteration & $ O \left( \left(|S|^2 |A|  +\frac{|S| |A|}{(1 - \gamma)^3} \right) \log\left(\frac{M}{\epsilon\delta}\right) \right)$ & This Paper
		\\  \hline
		Sublinear Randomized Value Iteration & $\otilde \left( \frac{|S| |A| M^2}{(1-\gamma)^4 \epsilon^2}\right) $ if $\epsilon \in (0, \frac{M}{\sqrt{(1-\gamma)}}]$ & This Paper
		\\ \hline
	\end{tabular}
\end{center}
\caption{\textbf{Running Times to Compute $\epsilon$-Approximate Policies in DMDPs with High Probability}: In this table, $\S$ is the number of states, $\A$ is the number of actions per state, $\gamma \in(0,1)$ is the discount factor, $M$ is an upper bound on the absolute value of any reward, and $C$ is an upper bound on the ergodicity.}
	\label{tab:literature_runtime_apx}
\end{table}

\paragraph{Sampling-Based Methods for MDP} Q-learning was the first sampling-based method for MDP, initially developed for reinforcement learning. Q-learning methods are essentially sampling-based variants of value iteration. \cite{kearns1999finite} proved that phased Q-learning takes $\tO(\frac{\S\A}{\epsilon^2})$ sample transitions to compute an $\epsilon$-optimal policy, where the dependence on $\gamma$ is left unspecified. No runtime analysis is given explicitly for phased Q-learning in \cite{kearns1999finite}, because in reinforcement learning settings, samples are observed experiences, and so the cost of sampling is not easily quantified. 

There is a large body of work on sampling-methods for MDPs in the literature of reinforcement learning, see e.g., \cite{strehl2006pac,strehl2009reinforcement,lattimore2012pac,azar2012sample,azar2013minimax} and many others.
These works study learning algorithms that update parameters by drawing information from some oracle, where the sampling oracles and modeling assumptions vary. The focus of this research is usually the sample complexity of learning; they are typically not concerned with the explicit runtime complexity, since the cost of sampling and processing a sample transition is nebulous in applied reinforcement learning settings. Our randomized algorithm is related to reinforcement learning under a generative model that allows the agent to sample transitions conditioned on specified state-action pairs, for examples, see \cite{azar2012sample, kearns2002sparse, azar2013minimax}. 

One notable contribution is the work of \cite{azar2013minimax}, which introduces policy iteration and value iteration algorithms for generative models that achieve the optimal sample complexity for finding $\epsilon$-optimal value functions, rather than $\epsilon$-optimal policies, as well as the matching lower bound. Their algorithms also achieve optimal sample complexity and optimal runtime for finding $\epsilon$-optimal policies in a restricted $\epsilon$ regime.

A recent related work \cite{wang2017randomized} proposed a randomized mirror-prox method with adaptive transition sampling, which applies to a special saddle point formulation of the Bellman equation. They achieve a total runtime of $\tO ( \frac{\S^3 \A M^2}{(1-\gamma)^6\epsilon^2} )$ for the general DMDP and $\tO (C\frac{\S \A}{(1-\gamma)^4\epsilon^2}) $ for DMDPs that are ergodic under all possible policies, where $C$ is a DMDP-specific ergodicity measure. 

\paragraph{Summary} Table~\ref{tab:literature_runtime_exact} and Table~\ref{tab:literature_runtime_apx} summarize the best-known running-time complexity of solution methods for DMDP. The running-time complexity is in terms of the total number of arithmetic operations.  Compared to deterministic methods and sampling-based methods, our main result has the sharpest dependence on the input dimensions $\S$ and $\A$ and discount factor $\gamma$. In terms of lower bounds for the DMDP, \cite{chen2017lower} recently showed that the runtime complexity for any randomized algorithm is $\Omega(\S^2\A)$. In the case where each transition can be sampled in $\tO(1)$ time, \cite{chen2017lower} showed that any randomized algorithm needs $\Omega(\frac{\S\A}{\epsilon})$ runtime to produce an $\epsilon$-optimal policy with high probability. \cite{azar2013minimax} also shows a sample complexity lower bound of $\Omega(\frac{\S \A}{(1-\gamma)^3\epsilon^2})$ for finding optimal policies under the generative model. In both the general and restricted cases, our main result nearly matches the lower bounds in its dependence on $\S$ and $\A$.

\section{Preliminaries}
\label{sec:prelim}

We describe a DMDP by the tuple $(S, A, P, r, \gamma)$, where $S$ is the finite state space, $A$ is the finite action space, $P$ is the collection of state-action-state transition probabilities, $r$ is the collection of state-action rewards, and $\gamma \in (0, 1)$ is a discount factor. We use $p_{a}(i,j)$ to denote the probability of going to state $j$ from state $i$ when taking action $a$ and define $p_a(i) \in \R^\states$ with $p_a(i)_j \defeq p_{a}(i,j)$ for all $j \in \states$. We use $r_a(i)$ to denote the reward obtained from taking action $a \in \actions$ at state $i \in \states$ and assume that for some known $M > 0$ it is the case that all $r_a(i) \in [-M, M]$.

We make the assumption throughout that for any state $i \in \states$ and action $a \in \actions$ we can sample $j \in \states$ independently at random so that  $\Pr[j = k] = p_a(i,k)$ in expected $O(1)$ time.  
This is a natural assumption under standard arithmetic models of computation that can be satisfied by preprocessing the DMDP in linear, i.e. $O(|S|^2 |A|)$ time.\footnote{As discussed in the introduction if instead sampling required $O(\log |S|)$ time, which is easily achieved with $\otilde(|S|^2|A|)$ preprocessing, this would increase our running times by only a multiplicative $O(\log |S|)$ which would be hidden by the $\otilde(\cdot)$ notation, leaving running times unaffected. 
}
For further discussion of sampling schemes, see \cite{wang2017randomized}.

We also use $\vec{1}$ and $\vec{0}$ to denote the all ones and all zeros vectors, respectively. 

In the remainder of this section we give definitions for several fundamental concepts in DMDP analysis that we use throughout the paper.

\begin{definition}[Value Operator]
For a given DMDP the \emph{value operator} $T : \valuespace \mapsto \valuespace$ is defined for all $u \in \valuespace$ and $i \in \states$ by
\begin{equation}
T(u)_{i} = \max_{a \in \actions} \left[ r_{a}(i) + \gamma \cdot p_a(i)^\top u \right] ~, \label{eq:value_operator}
\end{equation}
and we let $\vopt$ denote the \emph{value of the optimal policy $\pi^*$}, which is the unique vector such that $T(\vopt) = \vopt$. 
\end{definition}

\begin{definition}[Policy] We call any vector $\pi \in \policyspace$ a \emph{policy} and say that the action prescribed by policy $\pi$ to be taken at state $i \in \states$ is $\pi_i$. We let $T_\pi$ denote the \emph{value operator associated with $\pi$} defined for all $u \in \valuespace$ and $i \in \states$ by
\[
[T_\pi(u)]_i = r_{\pi_i}(i) + \gamma \cdot p_{\pi_i}(i)^\top u ~,  
\] 
and we let $v_\pi$ denote the \emph{values of policy $\pi$}, which is the unique vector such that $T_\pi(v_\pi) = v_\pi$. 
\end{definition}

Note that $T_\pi$ can be viewed as the value operator for the modified MDP where the only available action from each state is given by the policy $\pi$. Note that this modified MDP is essentially just an uncontrolled Markov Chain.

\begin{definition}[$\epsilon$-optimality] We say values $u \in \valuespace$ are \emph{$\epsilon$-optimal} if $\| \vopt - u\|_{\infty} \leq \epsilon$ and we say a policy $\pi \in \policyspace$ is \emph{$\epsilon$-optimal} if $\| \vopt - v_{\pi}\|_{\infty} \leq \epsilon$, i.e. the values of the policy are $\epsilon$-optimal.
\end{definition}

\subsection{Value Iteration Facts}

We review basic facts regarding value iteration that we use in our analysis. These are well established in the literature and for more details, please see Appendix~\ref{sec:value_facts}.

The first fact is that the value operator $T$ is a contraction mapping, meaning that the operator brings value vectors closer together. This is key to establishing the convergence and correctness of value iteration:

\begin{restatable}[Contraction Mapping]{lemma}{lemcontract}
\label{lem:contract} For all values $u,v \in \valuespace$ we have that
$
\|T(u) - T(v)\|_\infty \leq \gamma \|u - v\|_\infty 
$ 
and consequently $\|T(u) - \vopt\|_\infty \leq \gamma \|u - \vopt\|_\infty$, where $v^*$ is the optimal value vector.
\end{restatable} 

The second fact is that we can bound how close values are to being the value of a policy by how much the value operator for that policy moves the values:

\begin{restatable}{lemma}{lempolbound}
\label{lem:pol_bound_simple}
For any policy $\pi \in \policyspace$ and values $u \in \R^\states$ it holds that $\|T_\pi(u)- v_\pi\|_\infty \leq \frac{\gamma}{1 - \gamma} \|T_\pi(u) - u\|_\infty$, where $v_\pi$ is the exact value vector for policy $\pi$. 
\end{restatable}

The last fact is that the value operator is monotonic in the sense that it preserves the property that one value vector may be larger than another entry-wise.

\begin{restatable}{lemma}{lemvalmon}
\label{lem:val_is_mon}
 If values $u, v \in \valuespace$ satisfy $u \leq v$ entry-wise, then $T(u) \leq T(v)$ entry-wise.
\end{restatable}

\section{DMDP Algorithms}
\label{sec:DMDP}

In this section, we present our algorithms and the corresponding analysis for solving DMDPs. We split the presentation of our algorithms and its analysis into multiple pieces as follows:

\begin{itemize}
\item Section~\ref{section:approximate_value_operator}: we present $\apxvalcode$, our randomized, approximate value iteration sub-routine.
\item Section~\ref{section:basic_approx_vi}: we introduce a simple randomized value iteration scheme, $\randomVI$, which approximates the value operator using $\apxvalcode$ for each iteration. We analyze its convergence and correctness guarantees.
\item Section~\ref{section:polylog_variance_reduced_randomized_vi}: we use $\randomVI$ to create a high precision randomized value iteration algorithm, $\highprecisionVI$, which returns an $\epsilon$-optimal value vector in nearly linear time. 
\item Section~\ref{section:poly_variance_reduced_randomized_vi}: we present a sublinear time randomized value iteration algorithm, $\sublinearRandomVI$, which returns an $\epsilon$-optimal value vector in sublinear time.
\item Section~\ref{sec:obtaining_policy}: we show how to compute an $\epsilon$-optimal policy using our $\epsilon$-optimal value vectors. 
\item Section~\ref{section:monotonicity}: we present a monotonic value operator, $\apxmonvalcode$, that further improves the runtime of our sublinear randomized value iteration algorithm for computing $\epsilon$-approximate policies.
\end{itemize}
\subsection{Approximate Value Operator}
\label{section:approximate_value_operator}

Here we introduce our main sub-routine for performing randomized value iterations, $\apxvalcode$ (See Algorithm~\ref{alg:approx_val_op}). The routine $\apxvalcode$ approximates the value operator by sampling. Instead of computing the value operator exactly, i.e.:
\begin{equation}
T(u)_{i} = \max_{a \in \actions} \left[ r_{a}(i) + \gamma \cdot p_a(i)^\top u \right]
\end{equation}
$\apxvalcode$ approximates $T(u)_{i}$ by estimating $p_a(i)^\top u$ via sampling and maximizing over the action space based on these estimates.

The sampling procedure we use, $\apxtranscode$ is given in Algorithm~\ref{alg:apx_trans}. Its concentration guarantees are standard, but for completeness we give the analysis and the statement of Hoeffding's inequality below. 

\begin{algorithm}[H]
\caption{Approximate Transition: $\apxtranscode(u, M, i, a, \epsilon, \failprob)$}
\label{alg:apx_trans}
\begin{algorithmic}[1]
\Require{Values $u \in \valuespace$ and scalar $M \geq 0$ such that $\|u\|_\infty \leq M$}
\Require{State $i \in \states$ and action $a \in \actions$}
\Require{Target accuracy $\epsilon \in (0,1)$ and failure probability $\failprob \in (0,1)$}
\State{set $m := \lceil\frac{2 M^2}{\epsilon^2} \ln(\frac{2}{\failprob})\rceil$}
\For{\texttt{each  $k \in [m]$}}
\State \texttt{choose $i_k \in \states$ independently with $\Pr[i_k = j] = p_a(i,j)$.}
\EndFor
\BState \Return sample average $S = \frac{1}{m}\sum_{k \in [m]} u_{i_k}$
\end{algorithmic}
\end{algorithm}

\begin{theorem}[Hoeffding's Inequality (\cite{hoeffding} Theorem 2)]
	\label{thm:hoeffding}
	Let $X_{1},...,X_{m}$ be independent real valued random variables with $X_{i}\in[a_{i},b_{i}] $ for all $i\in[m]$ and let $Y=\frac{1}{m}\sum_{i\in[m]}X_{i}$. For all $t\geq 0$,
	\[
	\Prob\left[\left|Y-\mathbb{E}[Y]\right|\geq t\right]\leq2\exp\left(\frac{- 2m^{2}t^{2}}{\sum_{i\in[m]}(b_{i}-a_{i})^{2}}\right) .
	\]
\end{theorem}

\begin{lemma}[Sample Concentration]
\label{lem:trans_sample} 
$\apxtranscode$ (See Algorithm~\ref{alg:apx_trans}) can be implemented to run in time $O(M^2 \epsilon^{-2} \log(1/\failprob))$ and outputs $Y \in \R$ such that $|Y - p_a(i)^\top u| \leq \epsilon$ with probability $1 - \failprob$.
\end{lemma}

\begin{proof}
Note that $Y$ and $u_{i_k}$ were chosen so that $\E[Y] = p_a(i)^\top u$. Furthermore, since $u_{i_k} \in [-M, M]$ by the assumption that $\|u\|_\infty \leq M$ we have by  Hoeffding's Inequality, Theorem~\ref{thm:hoeffding}, and choice of $m$ that
\[
\Prob\left[
| Y - p_a(i)^\top u |
\geq \epsilon
\right]
=
\Prob\left[
| Y - \E Y |
\geq \epsilon
\right]
\leq 
2 \cdot \exp \left(
\frac{- 2 m^2 \epsilon^2}{m (2 M)^2}
\right)
\leq \failprob ~.
\]
Since the algorithm simply takes $O(m)$ samples and outputs their average the running time is $O(m)$.
\end{proof}

Using $\apxtranscode$ (Algorithm~\ref{alg:apx_trans}) our approximate value operator $\apxvalcode$ is given in Algorithm~\ref{alg:approx_val_op}. This algorithm works as previously described, using sampling to approximate $p_a(i)^\top u$ and then using these samples to approximately maximize over the space of actions. 

\begin{algorithm}
\caption{Approximate Value Operator: $\apxvalcode(u, v_0, x, \epsilon, \failprob)$}
\label{alg:approx_val_op}
\begin{algorithmic}[1]
	\Require{Current values, $u \in \valuespace$, and initial value: $v_0 \in \R^\states$}
	\Require{Precomputed offsets: $x \in \offsetspace$ with $|x_a(i) - p_a(i)^\top v_0| \leq \epsilon$ for all $i \in \states$ and $a \in \actions$.}
	\Require{Target accuracy $\epsilon \in (0,1)$ and failure probability $\failprob \in (0,1)$}
	\State{\texttt{compute} $M = \|u - v_0\|_\infty$}
	\For{\texttt{each state} $i \in \states$}
	\For{\texttt{each action $a \in \actions$}}
	\State \texttt{set  $\tilde{S}_a(i) = x_a(i) + \apxtranscode(u - v_0, M, i, a, \epsilon, \delta/(|\states||\actions|))$}
	\State \texttt{set $\tilde{Q}_a(i) = r_a(i) + \gamma \cdot \tilde{S}_a(i)$}
	\EndFor
	\State \texttt{set $\tilde{v}_i = \max_{a \in \actions} \tilde{Q}_a(i)$ and $\pi_i \in \argmax_{a \in \actions} \tilde{Q}_a(i)$}
	\EndFor
	\BState \Return $(\tilde{v},\pi)$\\ \emph{($\tilde{v} \in \valuespace$ is the result of an approximate value iteration and $\pi \in \policyspace$ is the corresponding policy )}.
\end{algorithmic}
\end{algorithm}

$\apxvalcode$ is designed to enable variance reduction. Instead of sampling values of $u$ directly and using the sample average, $\apxvalcode$ samples the difference from some fixed value vector input $v_0 \in \mathbb{R}^{S}$ and assumes that the value of $p_a(i)^\top v_0$ is approximately given through input $x_{i,a}$.
In particular, $\apxvalcode$ samples from $u - v_0$ and computes their average plus the offset $x_{i,a}$ which we denote by $S_a(i)$. Note that if we take $v_0$ to be the 0 vector, this is just standard naive sampling. However, by invoking $\apxvalcode$ with intelligently chosen $v_0$ and $x$ we are able to take less samples and still obtain precise estimates of $S_a(i)$. The fact we ensure $v_0$ is closer to $v$ and therefore the variance in our sampling is smaller and less samples are needed, is what we refer to as variance reduction and is crucial in achieving our best runtime results. 

In the following lemma we analyze $\apxvalcode$ (Algorithm~\ref{alg:approx_val_op})  and show that it approximates the value operator. The lemma is used repeatedly in later analysis.

\begin{lemma}[Approximate Value Operator]
\label{lem:aprox_value_operator}
$\apxvalcode$ (Algorithm~\ref{alg:approx_val_op}) can be implemented to run in time 
$$O\left(|S||A| \left\lceil \|u - v_0\|_\infty^2 \epsilon^{-2} \ln(|\states|  |\actions| / \delta) \right\rceil\right)$$ and with probability $1 - \failprob$ returns values $\tilde{v} \in \valuespace$ and policy $\pi \in \policyspace$ such that
$
\| \tilde{v} - T_\pi(u) \|_\infty \leq 2 \gamma \epsilon
$
and 
$
\| \tilde{v} - T(u)\|_\infty \leq 2 \gamma \epsilon
$.
\end{lemma}

\begin{proof}
The running time follows from the running time bound for $\apxtranscode$ given by Lemma~\ref{lem:trans_sample} and the facts that we can compute $\|u - v_0\|_\infty$ in $O(|\states|)$ time and perform the remaining arithmetic and maximization operations in $O(|\states| |\actions|)$ time. 
	
By applying union bound to $|\states| |\actions|$ applications of $\apxtranscode$ as analyzed by Lemma~\ref{lem:trans_sample} we know that with probability at least $1 - \failprob$ for all $i \in \states$ and $a \in \actions$ we have that 
\[
|\tilde{S}_a(i) - x_a(i) - p_a(i)^\top(u - v_0)| \leq \epsilon ~.
\]
By triangle inequality, in this case we have that for all $i \in S$ and $a \in A$
\begin{equation}\label{eqn:myeq}
|\tilde{S}_a(i) - p_a(i)^\top u| \leq 
|\tilde{S}_a(i) - x_a(i) - p_a(i)^\top(u - v_0)| + |x_a(i) - p_a(i)^\top v_0| \leq \epsilon + \epsilon \leq 2\epsilon ~.
\end{equation}
Now for all $i \in \states$ and $a \in \actions$ let 
\[
Q_a(i) \defeq r_a(i) + \gamma \cdot p_a(i)^\top u
 \text{ and }  
v_i \defeq \max_{a \in A} [Q_a(i)] ~.
\]
Note that $[T(u)]_i = v_i$ and $[T_\pi(u)]_i = Q_{\pi_i}(i)$. Now assuming \eqref{eqn:myeq} holds we know that 
\[
|Q_a(i) - \tilde{Q}_a(i)| \leq 2 \gamma \epsilon
\] for all $a \in A$, $i \in \states$. This implies 
\[
\|\tilde{v} - T_\pi(u)\|_\infty = \max_{i \in \states} |\tilde{Q}_{\pi_i}(i) - Q_{\pi_i}(i)| 
\leq 2 \gamma \epsilon ~. 
\]
Furthermore, this bound implies that for all $i \in \states$
\[
\tilde{v}_i = \max_{a \in \actions} \tilde{Q}_a(i)
\leq \max_{a \in \actions} Q_a(i) + 2\gamma\epsilon
= [T(u)]_i + 2\gamma\epsilon 
\]
and similarly, $\tilde{v}_i \geq [T(u)]_i - 2\gamma\epsilon$. Consequently, $|\tilde{v}_i - [T(u)]_i| \leq 2\gamma\epsilon$ and $\|\tilde{v} - T(u)\|_\infty \leq 2\gamma \epsilon$.
\end{proof}

In the following lemma we show that the guarantees of Lemma~\ref{lem:aprox_value_operator} suffice to prove that $\apxvalcode$ is approximately contractive. This connects the analysis of $\apxvalcode$ to standard value iteration and is key for establishing convergence and correctness for our algorithms. 

\begin{lemma}[Approximate Contraction]
\label{lem:approx_contraction}
If values $u,v \in \valuespace$ satisfy $\|v - T(u)\|_\infty \leq \alpha$ then \[
\|v - \vopt\|_\infty \leq \alpha + \gamma \|u - \vopt\|_\infty ~.
\]
\end{lemma}
\begin{proof}
By triangle inequality and the optimality of $\vopt$ we know that
\begin{align*}
\|v - v^*\|_\infty 
= \| v - \valop(v^*)\|_\infty 
\leq \| v - \valop(u) \|_\infty
+ \|\valop(u) - \valop(v^*)\|_\infty ~.
\end{align*}
Now, $\|v - \valop(u)\|_\infty \leq \alpha$ by assumption, and $\|\valop(u) - \valop(v^*)\|_\infty \leq \gamma \|u - v^*\|_\infty$ by the fact that $\valop$ is a $\gamma$-contraction map (See Lemma~\ref{lem:contract}). Combining yields the desired result. 
\end{proof}

\subsection{Basic Approximate VI Analysis}
\label{section:basic_approx_vi}

Here we introduce a simple randomized value iteration scheme, $\randomVI$, which approximates the value operator using $\apxvalcode$ in each iteration. We analyze its convergence and correctness guarantees here. 

\begin{algorithm}
\caption{Randomized VI: $\randomVI(v_{0}, \innerNumIter, \epsilon, \delta)$} \label{alg:randvi}
\begin{algorithmic}[1]
	\Require{Initial values $v_{0} \in \valuespace$ and number of iterations $\innerNumIter > 0$}
	\Require{Target accuracy $\epsilon \geq 0$ and failure probability $\failprob \in (0, 1)$}
	\State \texttt{Compute $x \in \offsetspace$ such that $x_a(i) = p_a(i)^\top v_{0}$ for all $i \in \states$ and $a \in \actions$}
	\For{\texttt{each iteration} $\innerIter \in [\innerNumIter]$}
	\State $(v_{\innerIter}, \pi_{\innerIter}) = \apxvalcode (v_{\innerIter - 1},v_{0},x,\epsilon,\delta/\innerNumIter)$
	\EndFor	
	\BState \Return $(v_{K}, \pi_{K})$
\end{algorithmic}
\end{algorithm}

\begin{lemma}[Quality of Randomized VI]
\label{lem:randomized_vi_analysis}
With probability $1 - \failprob$, an invocation of $\randomVI$ (Algorithm~\ref{alg:randvi}) produces $v_K \in \valuespace$ and $\pi_K \in \actions^\states$
such that
\[
\|v_\innerNumIter - \vopt\|_\infty \leq \frac{2 \epsilon \gamma}{1 - \gamma} + \exp(-\innerNumIter(1 - \gamma)) \cdot \|v_0 - \vopt\|_\infty ~.
\]
Consequently, if $\innerNumIter \geq \lceil\frac{1}{1 - \gamma} \log(\|v_0 - \vopt\|_\infty / (2 \epsilon \gamma)) \rceil$ then $\|v_\innerNumIter - v^*\|_\infty \leq \frac{4 \epsilon\gamma}{1 - \gamma}$.
\end{lemma}
\begin{proof}
By the analysis of $\apxvalcode$, Lemma~\ref{lem:aprox_value_operator}, and union bounding over its $\innerNumIter$ invocation, we have that with probability $1 - \delta$, $\|v_\innerIter - T(v_{\innerIter - 1})\|_\infty \leq 2 \gamma \epsilon$ for all $\innerIter \in [\innerNumIter]$. In this case, by Lemma~\ref{lem:approx_contraction}, we have that each iteration of $\randomVI$ is approximately contractive with $\alpha = 2\epsilon\gamma$, i.e. 
\[
\|v_\innerIter - \vopt\|_\infty \leq 2 \epsilon \gamma + \gamma \|v_{\innerIter - 1} - \vopt\|_\infty
\]
for all $\innerIter \in [\innerNumIter]$. Consequently, for all $\innerIter \geq 0$ we have
\begin{align*}
\|v_{\innerIter} - v^*\|_\infty
&\leq \sum_{\innerIter \in [\innerNumIter]} \gamma^\innerIter \cdot 2\epsilon\gamma
+
\gamma^\innerIter \cdot \|v_0 - v^*\|_\infty
\leq \frac{2\epsilon\gamma}{1 - \gamma} + \exp(- \innerIter (1 - \gamma)) \cdot \|v_0 - v^*\|_\infty ~,
\end{align*} 
where we used that $\sum_{\innerIter \in [\innerNumIter]} \gamma^\innerIter \leq \sum_{\innerIter = 1}^{\infty} \gamma^\innerIter = \frac{1}{1 - \gamma}$ and $\gamma = (1 - (1 - \gamma)) \leq \exp(- (1 - \gamma))$.	

The final claim follows from the fact that if $\innerNumIter \geq \frac{1}{1 - \gamma} \cdot \log(\|v_0 - v^*\|_\infty / 2\epsilon\gamma)$ then
\[
\frac{1}{1 - \gamma} \cdot 2\epsilon\gamma + \exp(- \innerNumIter (1 - \gamma)) \cdot \|v_{0} - v^*\|_\infty
\leq \frac{1}{1 - \gamma} \cdot 2\epsilon\gamma + 2\epsilon\gamma \leq \frac{4\epsilon\gamma}{1- \gamma}  ~.
\]
\end{proof}

\begin{lemma}[Randomized VI Runtime]
\label{lem:randomized_vi_runtime}
$\randomVI$ (Algorithm~\ref{alg:randvi}) can be implemented to run in time
\[
O\left(|S|^2 |A| + \innerNumIter |S|  |A| \left[
\frac{\|v_{0} - v^*\|_\infty^2}{\epsilon^2} + \frac{1}{(1 - \gamma)^2}
\right] \log\left(\frac{|S| |A| \innerNumIter}{\delta}\right)
\right) ~.
\]
\end{lemma}
\begin{proof}
Note that $x$ can be computed naively in time $O(|\states|^2 |\actions|)$. To bound the remaining running time we appeal to Lemma~\ref{lem:approx_contraction}, with $\alpha = 2\epsilon\gamma$, as justified by Lemma~\ref{lem:aprox_value_operator}. We have $\|v_{\innerIter} - v^*\|_\infty \leq \frac{2 \epsilon \gamma}{1 - \gamma} +\|v_{0} - \vopt\|_\infty$ for all $\innerIter \in [0, \innerNumIter]$. Consequently by triangle inequality and $(a + b)^2 \leq 2 a^2 + 2 b^2$ for all $a,b \in \R$: 
\begin{align*}
\|v_\innerIter - v_0\|_\infty^2 & \leq \|v_\innerIter - v^* + v^* - v_0\|_\infty^2 
 \leq \left( \|v_\innerIter - v^*\|_\infty + \|\vopt -v_0\|_\infty \right)^2 
 \leq 2 \|v_k - v^*\|_\infty^2 + 2 \|\vopt-v_0\|_\infty^2
 \\
& \leq 2\Big(\frac{2 \epsilon}{1 - \gamma}\Big)^2 + 2 \Big(2\|v_0 - v^*\|\Big)^2 
 = \frac{8\epsilon^2}{(1 - \gamma)^2} + 8 \|v_0 - v^*\|_\infty^2
\end{align*}
The result then follows from the running time of $\apxvalcode$ given by Lemma~\ref{lem:aprox_value_operator}. 
\end{proof}

Note that this analysis applies immediately to classic approximate value iteration algorithms. In particular, this analysis can extend the one provided by \cite{kearns1999finite} for phased Q-learning, to provide a total runtime guarantee in the restricted setting where (as assumed) taking a sample costs $O(1)$  time in expectation. 

\begin{theorem}[Simple Randomized VI] With probability $1 - \delta$ for $\epsilon < \frac{M}{1 - \gamma}$\footnote{Note that if $\epsilon$ is larger than $\frac{M}{1 - \gamma}$, a naive solution value vector $\vec{0}$ would suffice, since $\|\vec{0} - \vopt\|_\infty \leq \epsilon$} an invocation of 
\[
\randomVI \left(
\vec{0} ~ , ~ \left\lceil \frac{1}{1 - \gamma} \log \left(\frac{2 M }{ (1 - \gamma)^2 \epsilon}\right) \right\rceil ~ , ~ \frac{(1 - \gamma) \epsilon}{4\gamma} ~ , ~ \delta
\right)
\]
(See Algorithm~\ref{alg:randvi}), produces $v_\innerNumIter \in \valuespace$ such that $\|v_\innerNumIter - \vopt\|_\infty \leq \epsilon$ and can be implemented to run in time 
\[
\otilde\left(
\frac{|S|  |A| M^2}{(1 - \gamma)^5 \epsilon^2} 
\log \left(\frac{1}{\epsilon} \right) \right)~.
\]
\end{theorem}

\begin{proof}
In this case $x = \vec{0}$ and therefore computing it does not contribute to the running time, i.e. we do not have an additive $|S|^2 |A|$ term (also by our assumption we can sample in $O(1)$ expected time from the transition function). To bound the rest of the running time, note that $\|\vopt\|_\infty \leq \frac{M}{1 - \gamma}$ and therefore $\|\vec{0} - \vopt\|_\infty \leq \frac{M}{1 - \gamma}$. Consequently, the result follows from Lemma~\ref{lem:randomized_vi_analysis} and Lemma~\ref{lem:randomized_vi_runtime}.
\end{proof}

In the next several sections we show how to improve upon this result.

\subsection{High Precision Randomized VI in Nearly Linear Time}
\label{section:polylog_variance_reduced_randomized_vi}
In this section we improve on $\randomVI$ to obtain our best running times for computing an $\epsilon$-optimal value vector for small values of $\epsilon$. This improvement stems from the insight that the distance between the current value vector and the optimal value $\vopt$ as $\randomVI$ progresses can be bounded by the approximate contraction property of $\apxvalcode$, the subroutine that $\randomVI$ uses for approximate value iteration. Hence, in later iterations when the value vector is quite close to the optimal value, a more efficient sample complexity is sufficient to achieve similar concentration guarantees. 

Our algorithm, $\highprecisionVI$, is quite simple. In each iteration the algorithm calls $\randomVI$ for a decreasing error requirement, $\epsilon_k$. In the first iteration when the distance between the current value vector and the optimal value is high, $\highprecisionVI$ simply decreases the distance to optimal value vector by a factor of $\frac{1}{2}$. Then $\highprecisionVI$ appeals to $\randomVI$ again using the new value vector as the initial vector input, knowing that because this new vector is much closer to $\vopt$, subsequent approximate value iterations will require fewer samples to achieve an error $\frac{1}{4}$ of the original. This process is repeated so that the cost of each iteration is less than the same upper bound but so that after $k$ iterations the error is a $1/2^k$ fraction of the original. Ultimately this allows $\randomVI$ to produce an $\epsilon$-optimal value vector in nearly linear time. 

\begin{algorithm}[H]
\caption{High Precision Randomized VI: $\highprecisionVI(\epsilon, \delta)$}
\label{alg:polylog_randomized_vi}
\begin{algorithmic}[1]
\Require{Target precision $\epsilon$ and failure probability $\failprob \in (0, 1)$}
\State \texttt{Let $K = \lceil \log_2(\frac{M}{\epsilon(1 - \gamma)}) \rceil$ and $\innerNumIter = \lceil\frac{1}{1 - \gamma} \log(\frac{4}{1 - \gamma})\rceil$}
\State \texttt{Let $v_0 = \vec{0}$ and $\epsilon_0 = \frac{M}{1 - \gamma}$}
\For{\texttt{each iteration $k \in [K]$}}
\State $\epsilon_k = \frac{1}{2} \epsilon_{k - 1} = \frac{M}{2^{k} (1 - \gamma)}$
\State $(v_k, \pi_k) = \randomVI(v_{k - 1},\innerNumIter,(1 - \gamma) \epsilon_k / (4 \gamma),\delta/K)$ 
\EndFor
\BState \Return $(v_K, \pi_K)$
\end{algorithmic}
\end{algorithm}

\begin{lemma}[Quality and Runtime of High Precision Randomized VI]
\label{lem:polylog_randomized_vi_quality}
With probability $1 - \delta$ in an invocation of
$\highprecisionVI$ (See Algorithm~\ref{alg:polylog_randomized_vi}) we have that $\|v_k - \vopt\|_\infty \leq \epsilon_k$ for all $k \in [0, K]$ and therefore $v_K$ is an $\epsilon$-optimal value vector. 
\end{lemma}
\begin{proof}
We prove by induction on $k$ for that all $k  \in [0, K]$ that $\|v_k - \vopt\|_\infty \leq \epsilon_k$ with $\epsilon_{0} = \frac{M}{2^0 (1 - \gamma)}$ provided that each iteration of $\randomVI$ meets the criteria of Lemma~\ref{lem:randomized_vi_analysis}. Since by union bound this happens with probability $1 - \delta$ and since $\epsilon_k \leq \epsilon$ this suffices to prove the claim. 

Note that clearly $v_0 = \vec{0}$ and therefore $\|v_0 - \vopt\|_\infty \leq \frac{M}{1-\gamma} = \epsilon_0$ and therefore the base case holds. Now suppose that $\|v_{k - 1} - \vopt\|_\infty \leq \epsilon_{k - 1}$ for some $k \in [K]$. By Lemma~\ref{lem:randomized_vi_analysis} and our assumption that the call to $\randomVI$ succeeds we have that so long as 
\[
\innerNumIter \geq \left\lceil\frac{1}{1 - \gamma} \log\left( \frac{\|v_{k - 1} - \vopt\|_\infty }{ (2 [(1 - \gamma) \epsilon_k / (4\gamma)]  \gamma)}\right) \right\rceil
= 
\left\lceil\frac{1}{1 - \gamma} \log\left( \frac{2 \|v_{k - 1} - \vopt\|_\infty }{ (1 - \gamma) \epsilon_k }\right) \right\rceil
\]
which it is as by the assumption that $\|v_{k - 1} - \vopt\|_\infty \leq \epsilon_{k - 1} = 2 \epsilon_k$, then
\[
\|v_{k} - v^*\|_\infty \leq \frac{4 [(1 - \gamma) \epsilon_{k} / (4\gamma)] \gamma}{1 - \gamma} = \epsilon_{k} ~.
\]
Consequently, the result follows by induction.
\end{proof}

\begin{lemma}[Runtime of High Precision Randomized VI]
\label{lem:polylog_randomized_vi_runtime}
$\highprecisionVI$ (See Algorithm~\ref{alg:polylog_randomized_vi}) can be implemented so that with probability $1 - \delta$ it runs in time 
\[
\runtimenearline  ~.
\]
\end{lemma}

\begin{proof}
Each iteration $k$ calls  $\randomVI(v_{k - 1},\innerNumIter,(1 - \gamma) \epsilon_k / (4 \gamma),\delta/K)$. By Lemma~\ref{lem:randomized_vi_runtime} we note that the running time of the $k$th call is
\[
O \left( |S|^2 |A|  + 
\innerNumIter |S| |A| 
\left[\frac{\|v_{k - 1} - \vopt\|_\infty}{[(1 - \gamma) \epsilon_k / (4 \gamma)]^2} + \frac{1}{(1 - \gamma)^2}\right]
\log \left(\frac{|S| |A| \innerNumIter}{\failprob}\right)
\right)
~.
\]
However, by Lemma~\ref{lem:randomized_vi_runtime} we know that with probability $1 - \delta$, $\|v_{k - 1} - \vopt\|_\infty \leq \epsilon_{k - 1} \leq \epsilon_k$ for all $k$ and consequently, the cost of each iteration $k$ is
\[
O \left( |S|^2 |A|  + \frac{|S| |A|}{(1 - \gamma)^3}  
\log \left(\frac{1}{1 - \gamma} \right) 
\log \left(\frac{|S| |A| \log(\frac{1}{1 - \gamma}) \log(\frac{M}{\epsilon (1 - \gamma)})}{(1 - \gamma) \delta}\right) \right) ~.
\]
 Aggregating over the $K = \lceil \log_2(\frac{M}{\epsilon(1 - \gamma)}) \rceil$ iterations yields the desired running time.
\end{proof}

Note that this proof used that the running time of $\randomVI$ in $\highprecisionVI$ depends on $\|v_{k - 1} - \vopt\|^2_\infty/\epsilon_k^2$. By decreasing both $\|v_{k - 1} - \vopt\|^2_\infty$ and $\epsilon_k^2$ at geometric rates we ensured the cost of $\randomVI$ remained the same, even though the accuracy it yielded improved. In the next section we show how such variance reduction can be applied to obtain faster sublinear time algorithms.

\subsection{Randomized VI in Sublinear Time}
\label{section:poly_variance_reduced_randomized_vi}
Here we present an algorithm, $\sampledRandomVI$ for computing an $\epsilon$-optimal value vector that runs in sublinear time whenever $\frac{1}{\epsilon^2} = \tilde{o}(\frac{|S|(1-\gamma)^4}{M^2})$. Our algorithm is similar to $\highprecisionVI$; the primary difference is that instead of computing the initial offset $x$ in $\apxvalcode$ exactly in $O(|S|^2|A| )$ time, we compute an an approximation, $\tilde{x}$, to $\epsilon$-accuracy of $x$ in $\ell_\infty$, ie $\|x - \tilde{x}\|_\infty \leq \epsilon$, by sampling in $O(|S| |A| \|v_0\|_\infty^2 / \epsilon^2)$ time. Variance reduction is then performed as in $\highprecisionVI$, to decrease the number of samples need to run the rest of the algorithm. 

\begin{algorithm}[H]
\caption{Sampled Randomized VI: $\sampledRandomVI(v_{0}, 
\innerNumIter, \epsilon, \delta)$} \label{alg:sampled_randvi}
\begin{algorithmic}[1]
	\Require{Initial values $v_{0} \in \valuespace$ and number of iterations $\innerNumIter > 0$}
	\Require{Target accuracy $\epsilon > 0$ and failure probability $\failprob \in (0, 1)$}
	\State Sample to obtain approximate offsets: $\tilde{x} \in \offsetspace$ with $|\tilde{x}_a(i) - p_a(i)^\top v_0| \leq \epsilon$ for all $i \in \states$ and $a \in \actions$
	\For{\texttt{each round} $\innerIter \in [\innerNumIter]$}
	\State $(v_{\innerIter}, \pi_{\innerIter}) = \apxvalcode (v_{\innerIter-1},v_{0},\tilde{x},\epsilon,\delta/\innerNumIter)$
	\EndFor	
	\BState \Return $(v_{\innerNumIter}, \pi_{\innerNumIter})$
\end{algorithmic}
\end{algorithm}
The analysis for  $\sampledRandomVI$ follows exactly from Lemmas~\ref{lem:randomized_vi_analysis} and ~\ref{lem:randomized_vi_runtime}, since $\apxvalcode$ was designed and analyzed for the case where $|\tilde{x}_a(i) - p_a(i)^\top v_0| \leq \epsilon$. This condition was trivially satisified for $\randomVI$ when we used $x_a(i) = p_a(i)^\top v_0$. 

\begin{algorithm}[H]
\caption{Sublinear Time Randomized VI: $\sublinearRandomVI(\epsilon, \delta)$}
\label{alg:poly_randomized_vi}
\begin{algorithmic}[1]
\Require{Target precision $\epsilon$ and failure probability $\failprob \in (0, 1)$}
\State \texttt{Let $K = \lceil \log_2(\frac{M}{\epsilon(1 - \gamma)}) \rceil$ and $\innerNumIter = \lceil\frac{1}{1 - \gamma} \log(\frac{4}{1 - \gamma})\rceil$}
\State \texttt{Let $v_0 = \vec{0}$ and $\epsilon_0 = \frac{M}{1 - \gamma}$}
\For{\texttt{each iteration $k \in [K]$}}
\State $\epsilon_k = \frac{1}{2} \epsilon_{k - 1} = \frac{M}{2^{k} (1 - \gamma)}$
\State $(v_k, \pi_k) = \sampledRandomVI(v_{k - 1},\innerNumIter,(1 - \gamma) \epsilon_k / (4 \gamma),\delta/K)$ 
\EndFor
\BState \Return $(v_K, \pi_K)$
\end{algorithmic}
\end{algorithm}

\begin{lemma}[Quality of Sublinear Randomized VI]
\label{lem:poly_randomized_vi_quality}
In an invocation of
$\sublinearRandomVI$ (See Algorithm~\ref{alg:poly_randomized_vi}) with probability $1 - \failprob$ we have that $\|v_k - \vopt\|_\infty \leq \epsilon_k$ for all $k \in [0, K]$ and therefore $v_K$ is an $\epsilon$-optimal value vector. 
\end{lemma}
\begin{proof}
We can analyze $\sampledRandomVI$ identically to how we analyzed $\randomVI$.  Therefore we can analyze $\sublinearRandomVI$ similarly as we analyzed $\highprecisionVI$ in Lemma~\ref{lem:polylog_randomized_vi_quality} and the result follows.
\end{proof}

We now turn our attention to the runtime of $\sublinearRandomVI$. 

\begin{lemma}[Runtime of Sampled Randomized VI]
\label{lem:sampled_randomized_vi_runtime}
An invocation of $\sampledRandomVI$ (See Algorithm~\ref{alg:sampled_randvi}) can be implemented to run in time
\[
O\left(|S| |A| \left[  \frac{ M^2}{(1 - \gamma)^2 \epsilon^2} + \innerNumIter  \left(
\frac{\|v_{0} - v^*\|_\infty^2}{\epsilon^2} + \frac{1}{(1 - \gamma)^2}
\right) \right] \log\left(\frac{|S| |A| \innerNumIter}{\delta}\right)
\right) ~.
\]
\end{lemma}
\begin{proof}
We first claim that an invocation of $\sampledRandomVI$ (See Algorithm~\ref{alg:sampled_randvi}) can be implemented to run in time
\[
O\left(\left[ \frac{ |S|  |A| \|v_0\|_\infty^2}{\epsilon^2} + \innerNumIter |S|  |A| \left(
\frac{\|v_{0} - v^*\|_\infty^2}{\epsilon^2} + \frac{1}{(1 - \gamma)^2}
\right) \right] \log\left(\frac{|S| |A| \innerNumIter}{\delta}\right)
\right) ~.
\]
This claims largely follows the runtime analysis of $\randomVI$, see Lemma~\ref{lem:randomized_vi_runtime}. The only difference in the runtime comes from the fact that instead of computing $x$ exactly (which previously took $O(|S|^2 |A|)$ time), here we only need to compute an approximation, $\tilde{x}$, to $\epsilon$-accuracy of $x$ in $\ell_\infty$, ie $\|x - \tilde{x}\|_\infty \leq \epsilon$ without increasing the failure probability. As we showed in Section~\ref{section:approximate_value_operator} this can be done by sampling (See Lemma~\ref{lem:trans_sample}) in $O(|S| |A| \|v_0\|_\infty^2 \log(|S||A|/\delta) / \epsilon^2)$ time. The result then follows as
\begin{align*}
\|v_0\|_\infty^2
&\leq 
\left(\|v_0 - v^*\|_\infty + \|v^*\|_\infty\right)^2
\leq 2 \|v_0 - v^*\|_\infty^2 + 2 \|v^*\|_\infty^2
\\
&\leq 2 \|v_0 - v^*\|_\infty^2 + 2 \frac{M^2}{(1 - \gamma)^2} ~.
\end{align*}
where we used that $\|v_*\|_\infty \leq \frac{M}{1 - \gamma}$ and $(x + y)^2 \leq 2 x^2 + 2 y^2$ for all $x,y \in \R$. 
\end{proof}

\begin{lemma}[Runtime of Sublinear Randomized VI]
\label{lem:sub_rand_vi_time}
$\sublinearRandomVI$ (See Algorithm~\ref{alg:poly_randomized_vi}) can be implemented to run in time 
\[
 \otilde \left( |S| |A| \left(
 \frac{M^2}{(1-\gamma)^4 \epsilon^2} + \frac{1}{(1 - \gamma)^3}
 \right)
 \log \left(\frac{1}{\delta}\right) \right)
 \]
\end{lemma}
\begin{proof}
The runtime analysis for $\sublinearRandomVI$ follows the skeleton of the runtime analysis of $\highprecisionVI$ in Lemma~\ref{lem:polylog_randomized_vi_runtime}.  Each iteration $k$ calls  $\sampledRandomVI(v_{k - 1},\innerNumIter,(1 - \gamma) \epsilon_k / 4 \gamma,\delta/K)$. Note that $\innerNumIter = \lceil\frac{1}{1 - \gamma} \log(\frac{4}{1 - \gamma})\rceil$, $\epsilon = (1 - \gamma) \epsilon_k / 4 \gamma$ and by Lemma~\ref{lem:sampled_randomized_vi_runtime} with probability $1 - \delta$ we have that $\|v_{k - 1} - \vopt\|_\infty \leq \epsilon_{k - 1} \leq 2\epsilon_{k}$. Substituting these terms into the runtime analysis prescribed by Lemma~\ref{lem:sampled_randomized_vi_runtime} yields that each iteration $k$ costs
\[
 \otilde \left( |S| |A| \left(
 \frac{M^2}{(1-\gamma)^4 \epsilon^2_k} + \frac{1}{(1 - \gamma)^3}
 \right)
 \log \left(\frac{1}{\delta}\right) \right)
\]
 Aggregating over $K = \lceil \log_2(\frac{M}{\epsilon(1 - \gamma)}) \rceil$ iterations and using that $k \leq \lceil \log_2(\frac{M}{\epsilon(1 - \gamma)}) \rceil$ and therefore $\epsilon_k = \Omega(\epsilon)$ for all $k \in [K]$ yields the desired running time. 
\end{proof}

\begin{remark}
\label{remark:monotonic_improvement}
In the next section, Section~\ref{sec:obtaining_policy}, we discuss the sub-optimality guarantees for the policies given by the algorithms in this section. The approach in Section~\ref{sec:obtaining_policy} incurs an additional penalty when rounding to a policy, but in Section~\ref{section:monotonicity}, we show how to modify our algorithms so that values increase monotonically and guarantees on values immediately apply to guarantees on policies. For the monotonic algorithm, the additive $1/(1 - \gamma)^2$ term in Lemma~\ref{lem:randomized_vi_runtime} and Lemma~\ref{lem:sampled_randomized_vi_runtime} can be removed, but it is unclear how to remove it in Lemma~\ref{lem:sub_rand_vi_time}. Consequently, this monotonic algorithm computes $\epsilon$-optimal values and a corresponding$\epsilon$-optimal policy in a runtime which matches the one stated in Lemma~\ref{lem:sub_rand_vi_time}. In the regime where $\epsilon \in (0, \frac{M}{\sqrt{(1 - \gamma)}}]$, it follows that $\frac{M^2}{(1 - \gamma)^4 \epsilon^2} \geq \frac{1}{(1 - \gamma)^3}$; consequently, we assume $\epsilon \in (0, \frac{M}{\sqrt{1-\gamma}}]$ when stating our sublinear runtime of $\otilde \left( |S| |A| \left(
 \frac{M^2}{(1-\gamma)^4 \epsilon^2}
 \right) \log \left(\frac{1}{\delta}\right) \right)$ to find an $\epsilon$-optimal policy. Note that $\|v^*\|_\infty$ can be as large as $\frac{M}{1-\gamma}$ and consequently, there is a non-trivial regime of $\epsilon \in (\frac{M}{\sqrt{1-\gamma}}, \frac{M}{1-\gamma})$ where  instead the additive $\frac{1}{(1- \gamma)^3}$ would dominate runtimes in the algorithms of this paper.

\end{remark}

\subsection{Obtaining a Policy}
\label{sec:obtaining_policy}
In this section, we discuss how to leverage the analysis in the previous section to compute approximately optimal policies for DMDPs. We provide fairly general techniques to turn the $O(\epsilon/(1 - \gamma))$-approximate value vectors computed by $\randomVI$ and $\sampledRandomVI$) into $O(\epsilon/(1 - \gamma)^2)$-approximate policy vectors. This yields our fastest nearly linear convergent algorithms for computing $\epsilon$-approximate policies. In the next section we show how to improve upon these techniques and obtain even faster sublinear time algorithms. 

We start with the following lemma that essentiall follows from Proposition 2.1.4 in \cite{bertsekas2013abstract}.
 
\begin{lemma}
\label{lem:policy_quality_naive}	
$\randomVI$ and $\sampledRandomVI$ with $\innerNumIter \geq 1 +  \frac{1}{1 - \gamma} \cdot \log(\|v_0 - v^*\|_\infty / 2\epsilon\gamma)$ produce a policy $\pi_\innerNumIter$ that is $16 \epsilon / (1 - \gamma)^2$-optimal with probability  $1 - \delta$.
\end{lemma}

\begin{proof}
By Lemma~\ref{lem:randomized_vi_analysis} we know that both 
\[
\|v_\innerNumIter - \vopt\|_\infty \leq \frac{4 \epsilon \gamma}{1 - \gamma}
\text{ and  }
\|v_{\innerNumIter-1} - \vopt\|_\infty \leq \frac{4 \epsilon \gamma}{1 - \gamma}
~.
\]
Furthermore, by Lemma~\ref{lem:aprox_value_operator} we know that 
\[
\|v_\innerNumIter - T_{\pi_{\innerNumIter}} (v_{\innerNumIter - 1})\|_\infty \leq 2 \gamma  \epsilon ~.
\]
Combining these facts and invoking Lemma~\ref{lem:pol_bound_simple} gives
\begin{align*}
\| v_{\pi_\innerNumIter} - \vopt\|_\infty
&\leq 
\| v_{\pi_\innerNumIter} - T_{\pi_\innerNumIter} (v_{\innerNumIter - 1})\|_\infty
+ \|T_{\pi_\innerNumIter} (v_{\innerNumIter - 1}) - v_\innerNumIter\|_\infty 
+ \|v_\innerNumIter - \vopt\|_\infty \\
&\leq 
\frac{\gamma}{1 - \gamma}
\| v_{\innerNumIter - 1} - T_{\pi_\innerNumIter} (v_{\innerNumIter - 1})\|_\infty
+ 2\gamma \epsilon 
+ \frac{4 \epsilon\gamma}{1 - \gamma} ~.
\end{align*}
Furthermore, we have that 
\begin{align*}
\| v_{\innerNumIter -1} - T_{\pi_\innerNumIter} (v_{\innerNumIter-1})\|_\infty
&\leq 
\| v_{\innerNumIter - 1} - v^*\|_\infty
+ \|v^* - v_\innerNumIter\|_\infty 
+ \|v_\innerNumIter - T_{\pi_\innerNumIter} (v_{\innerNumIter - 1})\|_\infty 
\\
&\leq 
\frac{4 \epsilon\gamma}{1 - \gamma}
+ \frac{4 \epsilon\gamma}{1 - \gamma}
+ 2  \gamma \epsilon
\end{align*}
Combining and using the fact that $\gamma \in (0, 1)$ and $1 / (1 - \gamma) \geq 1$ yields the result.
\end{proof}

This lemma allows us to immediately claim running times for computing $\epsilon$-approximate policies simply be computing $O(\epsilon(1 - \gamma))$-approximate values as before.  
The proofs of the following two corollaries are immediate from our previous analysis.

\begin{corollary}[High Precision Approximate Policy Computation]
	\label{lem:polylog_policy}
	$\highprecisionVI$ (See Algorithm~\ref{alg:polylog_randomized_vi}) can be implemented so that with probability $1 - \delta$ it computes an $\epsilon$-approximate policy
	in time 
	\[
\runtimenearline  ~.
	\]
\end{corollary}

\begin{corollary}[Sublinear Time Approximate Policy Computation]
	$\sublinearRandomVI$ (See Algorithm~\ref{alg:poly_randomized_vi}) can be implemented so that with probability $1 - \delta$ it computes an $\epsilon$-approximate policy in time 
	\[
	\otilde \left( |S| |A| \left(
	\frac{M^2}{(1-\gamma)^6 \epsilon^2} + \frac{1}{(1 - \gamma)^3}
	\right)
	\log \left(\frac{1}{\delta}\right) \right)
	\]	
\end{corollary}

In the next section we show how to greatly improve this sublinear running time for computing an $\epsilon$-approximate policy by presenting a modified version of $\sublinearRandomVI$ that uses a monotonic variant of $\apxvalcode$, called $\apxmonvalcode$, that avoids this $1 / (1-\gamma)^2$ loss (and the additive $1/(1 - \gamma)^3$ term).

\subsection{Improved Monotonic Algorithm}
\label{section:monotonicity}

Here we discuss how to improve upon the running times for obtaining policies achieved in Section~\ref{section:poly_variance_reduced_randomized_vi}. We show how to obtain $\epsilon$-approximate policies in the same time our algorithms obtained $O(\epsilon)$-approximate values, i.e. we show how to avoid the loss of $1/(1 - \gamma)$ factors in running times that occurred in Section~\ref{section:poly_variance_reduced_randomized_vi}.

Our key insight to obtain this improvements is given in Lemma~\ref{lem:monotonicity}. This lemma shows that if we obtain values $v \in \valuespace$ and a policy $\pi \in \policyspace$ such that $\valop_\pi(v) \geq v$ entry-wise, then $v \leq v_\pi \le \vopt$ entry-wise. In other words, if we maintain $\valop_\pi(v) \geq v$ and obtain values $v$ that are $\epsilon$-optimal then the corresponding policy $\pi$ is $\epsilon$-optimal. Consequently, to avoid loss in optimality when converting approximate values to approximate policies we simply need to ensure that $\valop_\pi(v) \geq v$ entrywise. 

In this section we show how to maintain such values and policies by modifying our algorithms to be \emph{monotonic}. That is we modify our routines to start with an under-estimate of the optimal values and only increase them monotonically during the course of the algorithm. This modification allows us to not only maintain a policy with the desired properties, but also to remove the extra factor of $1/(1 - \gamma)^2$ in the running time of $\randomVI$ (see Lemma~\ref{lem:randomized_vi_runtime}).

In the remainder of this section we prove the basic properties of the value operator we use, Lemma~\ref{lem:monotonicity}, show how to achieve our monotonic version of $\apxvalcode$, which we call $\apxmonvalcode$ (Algorithm~\ref{alg:mon_val_op}) and analyze it in Lemma~\ref{lem:apxmon_mon}, Lemma~\ref{lem:mon_apx_val}, and Lemma~\ref{lem:mon_val_time}. We conclude by giving our fastest known running times for computing $\epsilon$-approximate policies in Theorem~\ref{thm:sublin_policy}. This theorem analyzes the algorithm, $\sublinearRandomMonVI$ (see Algorithm~\ref{alg:sampled_randvi}), which is a monotonic version of $\sublinearRandomVI$. This algorithm invokes 
$\sampledRandomMonVI$ (see Algorithm~\ref{alg:mon_sampled_randvi}) which is a monotonic version of the algorithm, $\sampledRandomVI$.

First we present the lemma which proves the key properties of the value this operator that we leverage to obtain higher quality policies. This lemma essentially follows from Proposition 2.2.1 in \cite{bertsekas2013abstract}.

\begin{lemma}
\label{lem:monotonicity}
If $v \leq \tilde{v} \leq T_{\tilde{\pi}}(v)$ entry-wise for $v,\tilde{v} \in \valuespace$ and $\tilde{\pi} \in \policyspace$, then $v \leq \tilde{v} \leq v_{\tilde{\pi}} \leq \vopt$ entry-wise.
\end{lemma}

\begin{proof}
By Lemma~\ref{lem:val_is_mon} and the fact that $v \leq \tilde{v}$ we know that $T_{\tilde{\pi}}(v) \leq T_{\tilde{\pi}}(\tilde{v})$. Furthermore since $\tilde{v} \leq T_{\tilde{\pi}}(v)$ by assumption, we know that $\tilde{v} \leq T_{\tilde{\pi}}(\tilde{v})$. Applying $T_{\tilde{\pi}}$ preserves the inequality so we have $T_{\tilde{\pi}}(\tilde{v}) \leq T_{\tilde{\pi}}^2(\tilde{v})$ and by induction obtain that entry-wise,
$$\tilde{v} \leq T_{\tilde{\pi}}(\tilde{v}) \leq T_{\tilde{\pi}}^2(\tilde{v}) \ldots \leq T^{\infty}_{\tilde{\pi}}(\tilde{v}) = v_{\tilde{\pi}}$$
That $v_{\tilde{\pi}} \leq \vopt$ entry-wise follows trivially from that $\tilde{\pi}$ can yield values better than the optimum policy.
\end{proof}

Next, we present our monotonic approximate value operator, $\apxmonvalcode$ in Algorithm~\ref{alg:mon_val_op}. $\apxmonvalcode$ modifies the output of $\apxvalcode$ to obtain the invariants of Lemma~\ref{lem:monotonicity} - provided they hold for the input. The modifications are straightforward. First we subtract a small amount from the output of $\apxvalcode$ to obtain underestimates for the value operator with high probability. 
If an underestimate is higher then the current value estimate for a particular state, then we update the policy to the new action for that state and the new value, otherwise we maintain the same action for this state. This ensures that our values are always underestimates of the true values, the values are always increasing, and the invariants to invoke Lemma~\ref{lem:monotonicity} are maintained. Moreover, this does not significantly change the running time or the quality of the output in terms of how well it approximates the value operator. We prove these facts in the following lemmas.

\begin{algorithm}[H]
\caption{Monotonic Random Value Operator: $\apxmonvalcode(u, \pi, v_0, x, \epsilon, \failprob)$}
\label{alg:mon_val_op}
\begin{algorithmic}[1]
\Require{Current values, $u \in \valuespace$, current policy, $\pi \in \policyspace$ with $T_\pi(u) \geq u$, and initial value, $v_0 \in \valuespace$}
\Require{Precomputed offsets: $x \in \R^{\states \times \actions}$ with $|x_a(i) - p_a(i)^\top v_0| \leq \epsilon$ for all $i \in \states$ and $a \in \actions$.}
\Require{Target accuracy $\epsilon \in (0,1)$ and failure probability $\failprob \in (0,1)$}
\State{$(q,w) := \apxvalcode(u,v_0,x,\epsilon,\failprob)$} 
\For{\texttt{each} $i \in \states$}
\If{$q_i - 2 \gamma \epsilon > u_i$}
\State{$\tilde{v}_i = q_i - 2 \gamma \epsilon$}
\State{$\tilde{\pi}_i = w_i$}
\Else
\State{$\tilde{v}_i = u_i$}
\State{$\tilde{\pi}_i = \pi_i$}
\EndIf
\EndFor
\BState \Return $(\tilde{v},\tilde{\pi}) = \apxmonvalcode(u, v_0, x, \epsilon, \failprob)$\\ 
\emph{(Note: $\tilde{v} \in \valuespace$ is the result of an approximate value iteration that chooses between the best under estimated action and the previous best underestimate, and $\tilde{\pi} \in \policyspace$ is the corresponding policy)}.
\end{algorithmic}
\end{algorithm}

We now show that $\apxmonvalcode$ satisfies the monotonicity properties
 of Lemma~\ref{lem:monotonicity}.
\begin{lemma}[$\apxmonvalcode$ Monotonicity Properties]
\label{lem:apxmon_mon}
With probability $1 - \failprob$, an invocation of $\apxmonvalcode$ (Algorithm~\ref{alg:mon_val_op}) returns $\tilde{v} \in \valuespace$ and $\tilde{\pi} \in \policyspace$ such that 
$
u \leq \tilde{v} 
\leq T_{\tilde{\pi}}(u)
\leq T_{\tilde{\pi}}(\tilde{v})
$
entry-wise. 
\end{lemma}
\begin{proof}
By design  $v \leq \tilde{v}$ entry-wise trivially. To show that $\tilde{v}_i \leq T_{\tilde{\pi}}(\tilde{v})_i$ for all $i \in \states$ note that by Lemma~\ref{lem:aprox_value_operator} with probability $1 - \delta$ we have $\| q - T_w(u)\|_\infty \leq 2 \gamma \epsilon$ and therefore $-2\gamma \epsilon \leq \valop_{w}(u)_i - q_i \leq 2 \gamma \epsilon$ for all $i \in \states$.  Consequently, in the case where $q_i - 2 \gamma \epsilon > u_i$, we have: 
\[
\tilde{v}_i = q_i - 2\gamma \epsilon
\leq \valop_w(u)_i = \valop_{\tilde{\pi}}(u)_i
\]  
and in the case where $q_i - 2\gamma \epsilon \leq u_i$
\[
\tilde{v}_i = u_i \leq \valop_\pi(u)_i = \valop_{\tilde{\pi}}(u)_i
\] 
by assumptions. In either case we have that $\tilde{v_i} \leq T_{\tilde{\pi}}(u)_i$. However, since $u \leq \tilde{v}$ entry-wise, by Lemma~\ref{lem:val_is_mon} we have $T_{\tilde{\pi}}(u)_i \leq T_{\tilde{\pi}}(\tilde{v})_i$ as desired.
\end{proof}

\begin{lemma}[$\apxmonvalcode$ Approximates Value Operator]
\label{lem:mon_apx_val}
With probability $1 - \failprob$, an invocation of $\apxmonvalcode$ (Algorithm~\ref{alg:mon_val_op}), returns $\tilde{v} \in \valuespace$ and $\tilde{\pi} \in \policyspace$ with 
$
\| \tilde{v} - T_{\tilde{\pi}}(u) \|_\infty \leq 4 \gamma \epsilon
$
and 
$
\| \tilde{v} - T(u) \|_\infty \leq 4 \gamma \epsilon
$.
\end{lemma}
\begin{proof}
By Lemma~\ref{lem:aprox_value_operator} with probability $1 - \delta$ we have $\| q - T_w(u) \|_\infty \leq 2 \gamma \epsilon$ and $\| q - T(u) \|_\infty \leq 2 \gamma \epsilon$. 

Now if for any $i \in \states$ we have  $q_i - 2 \gamma \epsilon > u_i$ then 
\[
| \tilde{v}_i - T_{\tilde{\pi}}(u)_i | = |q_i - 2\gamma\epsilon - T_w(u)_i| \leq \|q - T_w(u)\|_\infty + 2\gamma\epsilon = 4\gamma\epsilon~.
\] 
and similarly
\[
| \tilde{v}_i - T(u)_i | = |q_i - 2\gamma\epsilon - T(u)_i| \leq \|q - T(u)\|_\infty + 2\gamma\epsilon = 4\gamma\epsilon~.
\]

On the other hand, if for any $i \in \states$ we have  $q_i - 2 \gamma \epsilon < u_i$, 
we note that since $\| q - T(u) \|_\infty \leq 2 \gamma \epsilon$ it follows that: 
\[
T_{\pi}(u)_i \leq T(u)_i \leq q_i + 2\gamma\epsilon \leq u_i + 4\gamma\epsilon~.
\]
Since $u_i \leq T_{\pi}(u)_i$ by assumption, this implies that  $| \tilde{v}_i - T_{\tilde{\pi}}(u)_i | \leq 4 \gamma \epsilon$ and $
| \tilde{v}_i - T(u)_i | \leq 4 \gamma \epsilon
$. 
\end{proof}

\begin{lemma}[Runtime of $\apxmonvalcode$]
\label{lem:mon_val_time}
$\apxmonvalcode$ (Algorithm~\ref{alg:mon_val_op}) can be implemented to run in time
\[
O\left(|S||A| \left\lceil\frac{\|v - v_0\|_\infty^2}{\epsilon^2} \ln\left(\frac{|\states| |\actions|}{\failprob}\right) \right\rceil\right) ~.
\]
\end{lemma}
\begin{proof}
The routine $\apxmonvalcode$ can be implemented in the time needed to invoke $\apxvalcode$ plus an additional $O(|S|)$ work. The running time therefore follows from the analysis of $\apxvalcode$ in Lemma~\ref{lem:aprox_value_operator}.
\end{proof}

We now have everything we need to obtain our improved sublinear time algorithm for computing approximate policies. In the remainder of this section we provide the algorithms that achieve this and analyze them to prove Theorem~\ref{thm:sublin_policy}.

\begin{algorithm}[H]
\caption{Monotonic Sampled Randomized VI: $\sampledRandomMonVI(v_{0}, \pi_{0}, T, \epsilon, \delta)$} \label{alg:mon_sampled_randvi}
\begin{algorithmic}[1]
	\Require{Initial values $v_{0} \in \valuespace$, initial policy, $\pi_0 \in \policyspace$ with $T_{\pi_{0}}(v_0) \geq v_0$, and number of iterations $T > 0$}
	\Require{Target accuracy $\epsilon > 0$ and failure probability $\failprob \in (0, 1)$}
	\State Compute approximate offsets: $\tilde{x} \in \offsetspace$ with $|\tilde{x}_a(i) - p_a(i)^\top v_0| \leq \epsilon$ for all $i \in \states$ and $a \in \actions$
	\For{\texttt{each round} $t \in [T]$}
	\State $(v_{t}, \pi_{t}) = \apxmonvalcode (v_{t-1},\pi_{t - 1}, v_{0},\tilde{x},\epsilon/2,\delta/T)$
	\EndFor	
	\BState \Return $(v_{T}, \pi_{T})$
\end{algorithmic}
\end{algorithm}

\begin{algorithm}[H]
\caption{Monotonic Sublinear Time Randomized VI: $\sublinearRandomMonVI(\epsilon, \delta)$}
\label{alg:monotonic_poly_randomized_vi}
\begin{algorithmic}[1]
\Require{Target precision $\epsilon$ and failure probability $\failprob \in (0, 1)$}
\State \texttt{Let $K = \lceil \log_2(\frac{M}{\epsilon(1 - \gamma)}) \rceil$ and $T = \lceil\frac{1}{1 - \gamma} \log(\frac{4}{1 - \gamma})\rceil$}
\State \texttt{Let $v_0 = \vec{0}$, $\pi_0 \in \actions^\states$ arbitrary, and $\epsilon_0 = \frac{M}{1 - \gamma}$}
\For{\texttt{each iteration $k \in [K]$}}
\State $\epsilon_k = \frac{1}{2} \epsilon_{k - 1} = \frac{M}{2^{k} (1 - \gamma)}$
\State $(v_k, \pi_k) = \sampledRandomMonVI(v_{k - 1}, \pi_{k - 1}, T,(1 - \gamma) \epsilon_k / (4 \gamma),\delta/K)$ 
\EndFor
\BState \Return $(v_K, \pi_K)$
\end{algorithmic}
\end{algorithm}

\begin{theorem}[Sublinear Time Approximate Policy Computation]
\label{thm:sublin_policy}
$\sublinearRandomMonVI$ (see Algorithm~\ref{alg:monotonic_poly_randomized_vi}) can be implemented so that with probability $1 - \delta$ it yields an $\epsilon$-approximate policy for $\epsilon \in (0, \frac{M}{\sqrt{1-\gamma}}]$ in time 
\[
\runtimesublin ~.
\]
\end{theorem}
\begin{proof}
The algorithm $\sublinearRandomMonVI$ is the same as the algorithm $\sublinearRandomVI$, with the exception that $\pi_k \in \policyspace$ are maintained and $\sampledRandomMonVI$ is used instead of $\sampledRandomVI$. Furthermore, the algorithm $\sampledRandomMonVI$ is the same as the algorithm $\sampledRandomVI$ with the exception that $\pi_k \in \policyspace$ are maintained, $\apxmonvalcode$ is used instead of $\apxvalcode$, and the value of $\epsilon$ used is decreased by a factor of two (to account for the slightly larger error of $\apxmonvalcode$ in Lemma~\ref{lem:mon_apx_val}).

Furthermore, since for $v_0 = \frac{-M}{1 - \gamma} \vones$, any policy $\pi \in \policyspace$, and any $i \in \states$ we have
\[
[T_{\pi}(v_0)]_i
= r_{\pi_i}(i) + \gamma \cdot p_{\pi_i}(i)^\top v_0
\geq -M + \gamma \cdot \frac{-M}{1 - \gamma}
= \frac{- M}{1 - \gamma}
\]
Consequently, for $v_0 \leq T_{\pi_0}(v_0)$ entry-wise for the $v_0$ and $\pi_0$ in $\sublinearRandomMonVI$ and therefore we have that $T_{\pi_k}(v_k) \geq v_k$ is maintained throughout these algorithms, with probability $1-\delta$. 

This implies that these algorithms can be analyzed exactly as they were previously (Lemma~\ref{lem:sub_rand_vi_time} and Lemma~\ref{lem:poly_randomized_vi_quality}), though now the $v_k$ increase monotonically throughout the algorithm and  $T_{\pi_k}(v_k) \geq v_k$ is maintained. Consequently, the additive $1/(1 - \gamma)^2$ term in the previous analysis (Lemma~\ref{lem:randomized_vi_runtime} and Lemma~\ref{lem:sampled_randomized_vi_runtime}) does not occur as $\|v_k - \vopt\|_\infty$ decreases monotonically with $k$ and we have that the policy returned is of the same quality as the value return with probability $1 - \failprob$. Further, since we assume $\epsilon \in (0, \frac{M}{\sqrt{1-\gamma}}]$ the additive $\frac{1}{(1 - \gamma)^3}$ term can be omitted (see Remark~\ref{remark:monotonic_improvement}).
\end{proof}

\section{Finite Horizon Markov Decision Process}
\label{section:finite_horizon_mdp}

In this section we show how to use the ideas and techniques developed in Section~\ref{sec:DMDP} to obtain faster runtimes for solving finite horizon Markov Decision Processes, a close relative of the infinite horizon DMDP. 

In the finite horizon Markov Decision problem, we are given a tuple $(S, A, P, r, H)$, where $S$ is the finite state space, $A$ is the finite action space, $P$ is the collection of state-action-state transition probabilities, $r$ is the collection of state-action rewards, and $H \in \mathbb{Z}^{+}$ is the discrete time horizon of the problem. We use $p_{a}(i,j)$ to denote the probability of going to state $j$ from state $i$ when taking action $a$ and define $p_a(i) \in \R^\states$ with $p_a(i)_j \defeq p_{a}(i,j)$ for all $j \in \states$. We use $r_a(i)$ to denote the reward obtained from taking action $a \in \actions$ at state $i \in \states$ and assume that for some known $M > 0$ it is the case that all $r_a(i) \in [-M, M]$.

The key difference between the finite horizon MDP and the infinite horizon discounted MDP is that in the former, a specified horizon $H$ is given, and in the latter, the given discount factor indirectly adds a time dimension to the problem (i.e. a smaller discount factor, meaning that future rewards are heavily discounted, suggests a shorter relevant time horizon). Rewards in the finite horizon model are typically undiscounted, however all results easily extend to the discounted case.

As with our results for DMDP, we make the assumption throughout that for any state $i \in \states$ and action $a \in \actions$ we can sample $j \in \states$ independently at random so that  $\Pr[j = k] = p_a(i,k)$ in expected $O(1)$ time. (See Section $\ref{sec:prelim}$ for further discussion).

The goal in solving a finite horizon MDP is to compute a non-stationary policy $\pi(i, h)$, for every $i \in S$ and $h \in \{0, 1, 2, \ldots, H-1\}$ that tells an agent which action to choose at state $i$ at time $h$ to maximize total expected reward over the entire time horizon $H$.

In the remainder of this section, we give 2 algorithms for computing $\epsilon$-optimal policies. The first algorithm is an adaptation of our simple randomized value iteration scheme, $\randomVI$, and the second one is an adaptation of our high precision randomized value iteration $\highprecisionVI$. The first algorithm is sublinear for large values of $\epsilon$. The high precision algorithm is competitive for smaller values of $\epsilon$, in the regime where the first algorithm would effectively have superlinear running time. 

\subsection {Randomized Value Iteration for Finite Horizon MDP}
The basic ideas in our simple randomized value iteration scheme, $\randomVI$, can used to generate an $\epsilon$-optimal policy for finite horizon MDPs. We make a few tweaks in our finite horizon extension, $\finiteHVI$. The first is that we output the policy for every time step, since the optimal policy for a finite horizon MDP is non-stationary. The second is that the error $\epsilon$-input to $\finiteHVI$ is the desired cumulative error guarantee for our policy. Also, we use descending iteration counts to be consistent with backwards induction, and without loss of generality, we start the backwards induction with an initial value vector $\vec{0}$, so the termination reward for ending at any state $i$ is $0$.

\begin{algorithm}[H]
\caption{Randomized Finite Horizon VI: $\finiteHVI(v_{0}, H, \epsilon, \delta)$} \label{alg:randomized_finite_horizon_VI}
\begin{algorithmic}[1]
	\Require{Initial values $v_{0} \in \valuespace$ and number of iterations $H > 0$}
	\Require{Target accuracy $\epsilon \geq 0$ and failure probability $\failprob \in (0, 1)$}
	\State \texttt{Compute $x \in \offsetspace$ such that $x_a(i) = p_a(i)^\top v_{0}$ for all $i \in \states$ and $a \in \actions$}
	\State Let $h=H$
	\While{$h>0$ }
	\State $(v_{h}, \pi_{h}) = \apxvalcode (v_{h+1},v_{0},x,\frac{\epsilon}{2 H},\delta/H)$
	\State Decrement $h$ by 1
	\EndWhile	
	\BState \Return $(v_{h}, \pi_{h})$ for each $h \in [H]$
\end{algorithmic}
\end{algorithm}

To start the backwards induction with initial value vector $\vec{0}$, we simply use $v_{0} = \vec{0}$ as the input to $\finiteHVI$. 
\begin{lemma}[Quality]
\label{lem:finitehvi_quality}
$\finiteHVI(\vec{0}, H, \epsilon, \delta)$ returns an $\epsilon$-optimal policy with probability $1-\delta$. 
\end{lemma}

\begin{proof}
Since the backwards induction of exact value iteration directly finds the optimal policy, the quality guarantee of $\finiteHVI$ comes directly from the guarantees of $\apxvalcode$, see Lemma~\ref{lem:aprox_value_operator}. At each iteration $h$, we allow an additive error of at most $\frac{\epsilon}{H}$ with failure probability $\frac{\delta}{H}$. The error for the non-stationary policy compounds additively in the finite horizon case, so summing up this additive error over $H$ iterations and taking a union bound yields an $\epsilon$-optimal policy with probability $1-\delta$.
\end{proof}

\begin{lemma} [Runtime]
\label{lem:finitehvi_runtime}
$\finiteHVI(\vec{0}, H, \epsilon, \delta)$ runs in time
\[
\tO\left(|S||A| \frac{M^2H^5}{\epsilon^2}\ln\left(\frac{|\states|  |\actions|H}{\delta}\right) \right)~.
\]
\end{lemma}

\begin{proof}
Since rewards are undiscounted (or equivalently, $\gamma = 1$), the range of the values in the vector for the approximate value iteration increases by a factor of $M$ in every iteration. More formally, using the results of Lemma~\ref{lem:aprox_value_operator}, at iteration $h$, and letting $t = H-h$, the invocation of $\apxvalcode$ runs in time
\[
\tO\left(|S||A| \frac{t^2M^2H^2}{\epsilon^2}\ln\left(\frac{|\states|  |\actions| H}{\delta}\right) \right)~.
\]
Summing over $H$ iterations yields a total runtime of 
\[
\tO\left(|S||A| \frac{M^2H^5}{\epsilon^2}\ln\left(\frac{|\states|  |\actions|H}{\delta}\right) \right)~.
\]
\end{proof}

\subsection{Variance Reduced Value Iteration for Finite Horizon MDP}
In this section we show how to apply variance reduction to more efficiently find an $\epsilon$-optimal non-stationary policy for the finite horizon MDP problem by using $\finiteHVI$ as a building block. The idea is to compute $p^T_a(i)v_{0}$ for a fixed $v_{0}$ and then approximate $p^T_a(i)(v_{k} - v_{0})$ to save on sample complexity. Our variance reduced algorithm $\varianceReducedFiniteHVI$ invokes $\finiteHVI$ multiple times using different proximal $v_{0}$ over horizon $H$. 
\begin{algorithm}[H]
\caption{Variance Reduced Finite Horizon VI: $\varianceReducedFiniteHVI(H, L, \epsilon, \delta)$} \label{alg:variance_reduced_finite_randomized_vi}
\label{alg:variance_reduced_finite_randomized_vi}
\begin{algorithmic}[1]
	\Require{Number of iterations $H > 0$, recompute threshold $L>0$}
	\Require{Target accuracy $\epsilon \geq 0$ and failure probability $\failprob \in (0, 1)$}
	\State $v_{0} = \vec{0}$
	\State $t = [H/L]$
	\While{$t>0$ }
		\State $(v_{tL + 0}, \pi_{tL+0}), \ldots, (v_{tL + L-1}, \pi_{tL + L-1}) = \finiteHVI(v_0, L, \frac{L\epsilon}{H}, \frac{L\delta}{H})$
		\State $v_{0} = v_{tL + 0}$
		\State Decrement $t$ by 1
	\EndWhile
	\BState \Return $(v_h, \pi_h)$ for each $h \in [H]$. 
\end{algorithmic}
\end{algorithm}

\begin{lemma}[Quality of $\varianceReducedFiniteHVI$] With probability $1 - \delta$ Algorithm~\ref{alg:variance_reduced_finite_randomized_vi} ($\varianceReducedFiniteHVI$) produces an $\epsilon$-optimal policy.
\end{lemma}

\begin{proof}
This policy error follows from the guarantees of $\finiteHVI$ given by  Lemma~\ref{lem:finitehvi_quality}. With failure probability $\frac{L\delta}{H}$, each invocation of $\finiteHVI$ returns an $\frac{L\epsilon}{H}$-optimal policy for the $L$-horizon problem using the last value vector from the previous invocation of $\finiteHVI$ as the termination value vector for the $L$-step subproblem. Taking union bound over the $H/L$ iterations yields an $\epsilon$-optimal policy with probability $1-\delta$. 
\end{proof}

\begin{lemma}[Runtime of $\varianceReducedFiniteHVI$] The running time of Algorithm~\ref{alg:variance_reduced_finite_randomized_vi} ($\varianceReducedFiniteHVI$) is
\[
\tO\left(|S|^2|A| \frac{H}{L} + |S||A| \frac{L^2M^2H^3}{\epsilon^2}\ln\left(\frac{|\states|  |\actions| H}{\delta}\right) \right)~.
\]
\end{lemma}

\begin{proof}
$\varianceReducedFiniteHVI$ invokes $\finiteHVI$ $\frac{H}{L}$ times. Each invocation of $\finiteHVI$, by Lemma~\ref{lem:finitehvi_runtime} takes $O(|S|^2|A|)$
 time to compute the $x \in \offsetspace$ such that $x_a(i) = p_a(i)^\top v_{0}$. The advantage that this affords is that the maximum range of the sampled random variables is $[-LM, LM]$ in each invocation of $\finiteHVI$. Therefore, each iteration within $\finiteHVI$, by appealing to Lemma~\ref{lem:aprox_value_operator} takes total time 
\[
\tO\left(|S||A| \frac{L^2M^2H^2}{\epsilon^2}\ln\left(\frac{|\states|  |\actions| H}{\delta}\right) \right)~.
\]
Aggregating over $L$ iterations in each invocation of $\finiteHVI$ yields a runtime of 
\[
\tO\left(|S||A| \frac{L^3M^2H^2}{\epsilon^2}\ln\left(\frac{|\states|  |\actions| H}{\delta}\right) \right)~.
\]
for each invocation of $\finiteHVI$. 
Now, $\finiteHVI$ is invoked exactly $\frac{H}{L}$ times, so that the running time of $\varianceReducedFiniteHVI$, without the cost of computing the offset vectors $x_a(i)$, is: 
\[
\tO\left(|S||A| \frac{L^2M^2H^3}{\epsilon^2}\ln\left(\frac{|\states|  |\actions| H}{\delta}\right) \right)~.
\]
When we account for the cost of computing the offset vectors $x_a(i)$, we conclude that the total runtime of $\varianceReducedFiniteHVI$ is 
\[
\tO\left(|S|^2|A| \frac{H}{L} + |S||A| \frac{L^2M^2H^3}{\epsilon^2}\ln\left(\frac{|\states|  |\actions| H}{\delta}\right) \right)~.
\]
\bigskip
\end{proof}

It is important to pick a good choice for $L$. First, note that $\varianceReducedFiniteHVI$ is an improvement over $\finiteHVI$ in the regime where the running time of $\finiteHVI$ is super linear in $|S|^2|A|$, since the recomputations in the variance reduced version, $\varianceReducedFiniteHVI$, incur a $|S|^2|A| \frac{H}{L}$ pre-processing runtime. So if $\epsilon$ is large enough that $\finiteHVI$ is sublinear in $|S|^2|A|$, it is better to choose $L = 1$ (which essentially degenerates to $\finiteHVI$). When $\epsilon$ is very small, then when invoking $\varianceReducedFiniteHVI$, we can pick $L$ to equilibrate between the two terms in the runtime. Equating the two expressions yields that $L^3 = \frac{\epsilon^2|S|}{H^2M^2}$. Therefore we conclude with the following theorem.

\begin{theorem}
$\varianceReducedFiniteHVI$, Algorithm~\ref{alg:variance_reduced_finite_randomized_vi} can return an $\epsilon$-optimal policy with probability $1-\delta$ in running time
\[
\tO\left(\frac{|S|^{5/3}|A| H^{5/3}M^{2/3}}{\epsilon^{2/3}}\ln\left(\frac{|\states|  |\actions| H}{\delta}\right) \right)~.
\]
\end{theorem}

\section{Summary and Remarks}

In this paper we developed a class of new value iteration-based algorithms that employ the variance reduction to achieve improved running times for solving DMDPs. These algorithms compute approximately optimal policies in nearly linear and even (sometimes) sublinear running times and improve upon the previous best known randomized algorithms for solving DMDP in terms of the dependence on $\gamma,\S$ and $\A$. We also extend our ideas to finite horizon MDP. For future research we hope to improve our algorithms' dependence on the discount factor. We also believe that our method can be generalized to settings when different structural knowledge is available, e.g., when there is a known upper bound on the diameter of the process, and when the process is known to be ergodic. We also plan to investigate the use of variance reduction techniques to speed up the policy iteration method and primal-dual methods.

\section{Acknowledgements}
We would like to thank Daochen Wang for pointing out an error in the previous version of our work regarding additive terms in our sublinear complexity bounds and the range of $\epsilon$ for which these results hold. We would like to thank Lin Yang for his generous help, meticulous proofreading, constructive suggestions, and useful discussions. We would also like to thank Moses Charikar for his support and encouragement in pursuing this research topic, and Ben Van Roy for excellent courses on Markov Decision Processes.
\bibliographystyle{alpha}
\bibliography{mdp,mdp2,nips_1705}

\appendix

\section{Value Iteration Facts}
\label{sec:value_facts}

Here we review some basic facts about value iteration we use in our analysis. These are well established in the MDP literature but we include it here for completeness.

\lemcontract*
Please see \cite{bertsekas2013abstract} and \cite{puterman2014markov} for more details on basic facts for DMDP.

\lempolbound*

\begin{proof} By the fact that $T_\pi(v_\pi) = v_\pi$ we have:
\begin{align*}
\|T_\pi(u) - v_\pi\|_\infty &= \|T_\pi(u) - T_\pi(v_\pi)\|_\infty \\
&= \|T_\pi(u) - T_\pi(T_\pi(u)) + T_\pi(T_\pi(u)) - T_\pi(v_\pi)\|_\infty \\
&\leq \|T_\pi(u) - T_\pi(T_\pi(u))\|_\infty + \|T_\pi(T_\pi(u)) - T_\pi(v_\pi)\|_\infty \\
&\leq \gamma \|T_\pi(u) - u\|_\infty + \gamma \|T_\pi(u) - v_\pi\|_\infty ~.
\end{align*}
Therefore, 
$(1-\gamma) \|T_\pi(u) - v_\pi\|_\infty \leq \gamma \|T_\pi(u) - u\|_\infty$
and so we conclude that 
\[
\|T_\pi(u) - v_\pi\|_\infty \leq \frac{\gamma}{1-\gamma} \|T_\pi(u) - u\|_\infty ~.
\]
\end{proof}

\lemvalmon*

\begin{proof}
Note that 
\[
T(u)_{i} = \max_{a \in \actions} \left[ r_{a}(i) + \gamma \cdot p_a(i)^\top u \right] \leq \max_{a \in \actions} \left[ r_{a}(i) + \gamma \cdot p_a(i)^\top v \right] = T(v)_{i} ~,
\]
where the inequality holds because $u \leq v$, and $\gamma$ and $p_a(i)$ are non-negative for all $a \in \actions, i \in \states$. 
\end{proof}

\newcommand{\OPT}{\mathrm{OPT}}

\newcommand{\mproj}{\textbf{P}}
\newcommand{\ma}{\textbf{A}}
\newcommand{\mE}{\textbf{E}}
\newcommand{\mM}{\textbf{M}}
\newcommand{\ms}{\textbf{S}}
\newcommand{\norm}[1]{\|#1\|}
\newcommand{\diag}{\mathrm{diag}}
\newcommand{\poly}{\mathrm{poly}}

\section{Solving DMDPs With Interior Point Methods}
\label{sec:dmdp_with_ip}

In this section, we show how to solve DMDPs using standard linear programming machinery, i.e. interior point methods. First, we provide a fairly standard formulation of a DMDP as a linear program and show that from an approximate solutions to the linear program an approximately optimal policy can be obtained. (See Definition~\ref{def:DMDPLP} and Lemma~\ref{lem:dmdp_lp_props}). While there is some loss in quality in converting an approximate DMDP LP solutions into policies, we show how to bound the loss so that it only only affects logarithmic factors in the final running time. 

We then show how to use interior point methods to solve this linear program. 
To make it easier to apply the interior point methods we reduce approximately solving this linear program to approximately solving $\ell_1$-regression, using a technique from \cite{LeeS15}. The fastest known interior point methods for solving $\ell_1$ regression, i.e. \cite{LeeS15}, gives the running time for solving DMDPS that we provided in Section~\ref{sec:prev_work} (See Theorem~\ref{thm:dmdp_with_ip})

We start by providing and analyzing our linear program formulation of the DMDP. 

\begin{definition}[DMDP Linear Program] 
\label{def:DMDPLP}
We call the following linear program the \emph{DMDP LP}
\begin{equation}
\label{eq:optimal_lp}
\begin{array}{ll@{}ll}
\text{minimize}   & \displaystyle v^\top \vones\\
\text{subject to}& \displaystyle \ma v \geq r\\
\end{array}
\end{equation}
where $r\in\R^{(S \times A)}$ is the vector of rewards, i.e. $r_{i,a} = r_a(i)$ for all $i\in S$ and $a\in A$ and  $\ma = \mE - \gamma \mproj$ where $\mE\in\R^{(S \times A)\times S}$ is the matrix where for all $i,j\in S$ and $a\in A$ we have that the $j$-th entry of row $(i,a)$ of $\mE$ is 1 if $i = j$ and 0 otherwise, and  $\mproj\in\R^{(S \times A)\times S}$ is a matrix where for all $i\in S$ and $a\in A$ we have that row $(i,a)$ of $\mproj$ is $p_a(i)$.

We call a vector $v \in \valuespace$ an \emph{$\epsilon$-approximate DMDP LP solution} if $\ma v \geq r - \epsilon \vones$ and $v^\top \vones \leq OPT + \epsilon$ where $OPT$ is the optimal value of the DMDP LP, \eqref{eq:optimal_lp}.
\end{definition}

This DMDP LP is a standard formulation of the DMDP as a linear program. In the next lemma we prove that solving the DMDP LP is equivalent to solving a DMDP and we show that approximately optimizing the objective of the DMDP LP while meeting the constraints corresponds to values that are approximately optimal in terms of local optimality of the value operator. 

\begin{lemma}[DMDP LP Properties]
\label{lem:dmdp_lp_props} 
In the DMDP LP a vector $v \in \valuespace$ satisfies $\ma v \geq r$ if and only if $v \geq T(v)$ entry-wise. 
If additionally $v$ satisfies $v^\top \vones \leq OPT + \epsilon$ then $T(v) \leq v \leq T(v) + \epsilon \vones$. 
Consequently, $\vopt$ the solution to the DMDP is the unique minimizer of the DMDP LP.
\end{lemma}

\begin{proof}
Note that $\ma v \geq r$ is equivalent to $v_{i} - \gamma\cdot p(i,a)^{\top}v \geq r_{i}(a)$ or  $v_{i} \geq r_{i}(a) +  \gamma\cdot p(i,a)^{\top}v$ for all $i \in \states$ and $a \in \actions$. Since $[T(v)]_i = \max_{a \in \actions} [r_{i}(a) +  \gamma\cdot p(i,a)^{\top}v]$ we have $\ma v \geq r$ if and only if $T(v) \leq v$.

Now suppose $v$ additionally satisfies $v^\top \vones \leq OPT + \epsilon$. Let $i \in [n]$ be a coordinate such that $[T(v) - v]_i$ is maximal. Let $\alpha = [v - T(v)]_i$ and let $w$ be the same vector as $v$ where coordinate $i$ is decreased by $\alpha$. Since $v \geq T(v)$ we know $\alpha$ is positive and see that $T(w) \geq w$ as well by our choice of $\alpha$ (since $T$ is monotonic in $v$). Consequently $\ma w \geq r$ and $w^\top \vones = v^\top \vones - \alpha \leq OPT + \epsilon - \alpha$. By the definition of OPT we therefore have that $\alpha \leq \epsilon$ and $v \leq T(v) + \epsilon \vones$.

Consequently, any solution $v_{lp}$ to the DMDP LP must satisfy $v_{lp} = T(v_{lp})$ and therefore be the unique optimal values of the DMDP.
\end{proof}

Using Lemma~\ref{lem:dmdp_lp_props} we show that given any $\epsilon$-approximate solution to the DMDP LP we can quickly compute an $O(\epsilon |S| (1 - \gamma)^{-2})$-approximate policy simply by computing the best actions with respect to these approximate values. 

\begin{lemma}[Approximate Policies from Approximate LP Solution]
\label{lem:apx_policy}
If $v$ is an $\epsilon$-approximate DMDP LP solution and if $\pi \in \policyspace$ is defined with $\pi_i = \argmax_{a \in \actions} r_a + \gamma \cdot p(i, a)^\top v$ for all $i \in \states$ then $\pi$ is an $8\epsilon |S| (1 - \gamma)^{-2}$-optimal policy. 
\end{lemma}

\begin{proof}
Note that $\ma \vones = (1 - \gamma) \vones$ and therefore $w \defeq v + \frac{\epsilon}{1 - \gamma} \vones$ satisfies $\ma w \geq r$. Since 
\[
w^\top \vones 
= v^\top \vones + \frac{\epsilon |S|}{(1 - \gamma)}
\leq OPT + \frac{2 \epsilon |S|}{(1 - \gamma)}
\]
Lemma~\ref{lem:dmdp_lp_props} yields 
$
T(w) \leq w \leq T(w) + 2C \vones
$
for $C \defeq \epsilon |S| (1 - \gamma)^{-1}$ and $\norm{w - T(w)}_\infty \leq 2 C$. Since clearly $\norm{v - w}_\infty \leq C$ and $\norm{T(v) - T(w)}_\infty \leq C$ then by triangle inequality we have that $\norm{v - T(v)}_\infty \leq 4 C$. 

Next, triangle inequality yields that 
$
\norm{v - \vopt} \leq \norm{T(v) - v} + \norm{T(v) - \vopt}
$
and since $T(\vopt) = \vopt$ by the contraction property we have $\norm{T(v) - \vopt}_\infty \leq \gamma \norm{v - \vopt}_\infty$ and combining these facts yields that $\norm{v - \vopt} \leq 4C \gamma (1 - \gamma)^{-1}$.
However, clearly by the definition of $\pi$ it is the case that $T(v) = T_\pi(v)$ and by Lemma~\ref{lem:pol_bound_simple} we have $\norm{T(v) - v_\pi} \leq \gamma (1 - \gamma)^{-1} \cdot 4C$. Consequently,
\[
\norm{v_\pi - \vopt}_\infty
\leq \norm{T(v) - v_\pi}_\infty + \norm{T(v) - \vopt}_\infty
\leq 8C\gamma(1 - \gamma)^{-1} 
\]
where we used that $\norm{T(v) - \vopt}_\infty \leq \gamma \norm{v - \vopt}_\infty \leq 4C\gamma(1 - \gamma)^{-1}$.
\end{proof}

With Lemma~\ref{lem:apx_policy} established, the goal in the rest of this section is to show how to produce an approximate DMDP LP solution quickly. While it may be possible to adapt the fastest generic linear programming algorithms (e.g. \cite{lee2015efficient}) directly to the DMDP LP formulation, for a simpler proof we reduce the original linear program to computing $\epsilon$-approximate $\ell_1$ regression. We do this so that we can easily invoke the following result of \cite{lee2015efficient} to obtain our running times.

\begin{corollary}[$\ell_1$ Regression Running Time - Corollary~24 of \cite{lee2015efficient}, Adapted] 
\label{cor:l1_regress}
Given $\mM \in \R^{n \times d}$ and $b \in \R^n$ then there is an algorithm which with high probability in time $\otilde(n d^{1.5} \log(1/\epsilon))$ finds $x \in \R^d$ with
\[
\|\mM x - b\|_1 \leq \min_{y \in \R^d} \|\mM y - b\|_1 + \epsilon \|b\|_1 ~.
\]
\end{corollary} 

To convert our DMDP LP to such an optimization problem we turn our constraint set into a symmetric one. First we modify the DMDP LP to have two sided constraints. We know that $\vopt$ satisfies $\|\vopt\|_\infty \leq \frac{M}{1 - \gamma}$ and therefore for all $i \in \states$ and $a \in \actions$ we have
\[
[\ma \vopt]_{(i,a)} = r_{(i, a)} + \gamma \cdot p_{(i, a)}^\top \vopt
\leq M  + \gamma \cdot \frac{M}{1 - \gamma}  
\leq \frac{2M}{1 - \gamma} ~.
\]
Furthermore, since we know that $\vopt$ is the optimizer of the DMDP LP by Lemma~\ref{lem:dmdp_lp_props} we therefore see that solving the DMDP LP is equivalent to minimizing $v^\top \vones$ under the constraint that $v \in P$ where 
\[
P=\left\{v \in \valuespace : 
r
\leq \ma v 
\leq \frac{2M}{1-\gamma}\vones
\right\}.
\]
Moreover, we can center this feasible region; for $s=\frac{1}{2}(\frac{2M}{1-\gamma}\vones - r)$
and $b=\frac{1}{2}(\frac{2M}{1-\gamma} \vones + r )$ we have 
\[
P=\{v \in \valuespace :-s\leq\ma v-b\leq s\} ~.
\]
Furthermore, if we let $\ms=\diag(s)$, i.e. the diagonal matrix with $s$ on the diagonal then 
\[
P=\left\{v \in \valuespace :\norm{\ms^{-1}\ma v-\ms^{-1}b}_{\infty}\leq1 \right\}
\]

Consequently, we wish to minimize $v^\top \vones$ under the constraint that $\norm{\ms^{-1}\ma v-\ms^{-1}b}_{\infty}\leq 1$. To turn this into a $\ell_1$ regression problem we use a technique from \cite{lee2015efficient} that $\ell_\infty$ constraints can be turned into $\ell_1$ objectives through the following simple lemma.

\begin{lemma}
\label{lem:l1_trick}
For all vectors $x \in R$ we have
\[
|x - 1| + |x + 1|
= 2 \cdot \max \{|x| , 1\}
\]
\end{lemma}

\begin{proof}
If $x \in [0, 1)$ then 
\[
|x - 1| + |x + 1| = 1 - x + 1 + x = 1 
\]
and if $x \notin [1, \infty)$ then 
\[
|x - 1| + |x + 1| = |x| - 1| + |x| + 1 = 2 |x|
\]
The result follows by symmetry for the case when $x < 0$.
\end{proof}

From this lemma we see that to penalize infeasibility, ie $v$ such that $\|\ms^{-1} \ma v - \ms^{-1} b\|_\infty > 1$, we can simply add a term of the form 
 $\|\ms^{-1} \ma v - \ms^{-1} b - \vones\|_1 + \|\ms^{-1} \ma v - \ms^{-1} b + \vones\|_1$ to the objective. We give the following $\ell_1$ regression problem:

\begin{definition}[DMDP $\ell_1$ Regression]
For a given DMDP and a parameter $\alpha$ (to be set later) we call the following $\ell_1$ regression problem the \emph{DMDP $\ell_1$ problem}\footnote{This equality below is an immediate consequence of Lemma~\ref{lem:l1_trick}}
\begin{align}
\label{eq:new_formulation_lp}
\min_{v} f(v) 
&= \left| \alpha \left(\frac{\S M }{1 - \gamma} + \vones^\top v \right) \right|
+ \left\| \ms^{-1} \ma v - \ms^{-1} b - \vones \right\|_1
+ \left\|\ms^{-1} \ma v - \ms^{-1} b + \vones \right\|_1
\\
&= \left| \alpha \left(\frac{\S M }{1 - \gamma} - \vones^\top v\right) \right| + 2\sum_{i \in \states} \max  \left\{
|
[\ms^{-1} \ma v - \ms^{-1} b]_i
| , 1
\right\}
\end{align}
where $s=\frac{1}{2}(\frac{2M}{1-\gamma}\vones - r)$
and $b=\frac{1}{2} (\frac{2M}{1-\gamma} \vones + r )$, and $\ms = \diag(s)$. We let $\vopt_f$ denote the optimal solution to this $\ell_1$ regression problem and we call $v$ an $\epsilon$-optimal solution to $f$ if $f(v) \leq f(\vopt_f) + \epsilon$.
\end{definition}

Note that our $\epsilon$-approximate solution to \eqref{eq:new_formulation_lp} may be infeasible for the original linear program, \eqref{eq:optimal_lp}, and $ 2\sum_{i} \max \left\{
| 
[\ms^{-1} \ma v - \ms^{-1} b]_i | , 1
\right\}$ characterizes the magnitude of our constraint violation for the original MDP problem, $\eqref{eq:optimal_lp}$. However, by carefully choosing $\alpha$ we can show that we can trade-off violating these constraints with lack of optimality in the objective and therefore $\epsilon$-approximate solutions to \eqref{eq:new_formulation_lp} still give approximate DMDP-LP solutions. We quanitfy this below.

\begin{lemma}
\label{lem:l1_regress_qual}
Suppose that $v$ is an $\epsilon$-approximate solution to the DMDP $\ell_1$ problem then $v$ is an $\epsilon'$-approximate DMDP LP solution for
\[
\epsilon' \leq \max \left\{
\frac{\epsilon}{\alpha}
~ , ~
\frac{2 \alpha |S| M^2}{(1-\gamma)^2} 
+ \frac{\epsilon M}{(1 - \gamma)}
\right\} ~.
\]
\end{lemma}

\begin{proof}
Recall that $
|
[\ms^{-1} \ma \vopt - \ms^{-1} b]_i
| \leq 1$ for all $i \in \states$ and  $\vones^\top \vopt \in [\frac{- |S| M}{1 - \gamma}, \frac{|S| M}{1 - \gamma}]$. Consequently,
\[
f(\vopt_f) \leq f(\vopt) = \alpha \left(
\frac{\S M }{1 - \gamma} + \vones^\top v^* \right) + 2\S
\leq \frac{2 \alpha \S M }{1 - \gamma}  + 2\S
\]
Since, $f(v) \leq f(\vopt_f) + \epsilon$ by assumption we have that
\[
\alpha \left(\frac{\S M }{1 - \gamma} + \vones^\top v\right)
+ 2 \S
\leq f(v) \leq \alpha \left(\frac{\S M }{1 - \gamma} + \vones^\top \vopt\right) + 2 \S + \epsilon
\]
and therefore $\vones^\top v \leq \vones^\top \vopt + \frac{\epsilon}{\alpha}$. Furthermore this implies,
\[
2\sum_{i \in \states} \max  \left\{
|
[\ms^{-1} \ma v - \ms^{-1} b]_i
| , 1
\right\} \leq f(v) \leq \frac{2 \alpha \S M}{1 - \gamma} + 2 \S + \epsilon 
\]
Now letting $\delta_1 = \epsilon + \frac{2 \alpha |S| M}{1 - \gamma}$ we see this in turn this implies that $\norm{\ms^{-1} \ma v - \ms^{-1} b}_\infty \leq \frac{\delta_1}{2}+1$. Consequently, $\ms^{-1} \ma v - \ms^{-1} b \geq -(\frac{\delta_1}{2}+1) \vones$ and  
\[
\ma v 
\geq b - s -  \frac{\delta_1}{2}s 
\geq r - \frac{\delta_1}{4}r -  \frac{\delta_1}{2}\left(\frac{M}{1 - \gamma} \vones\right)
\geq r - \frac{\delta_1 M}{(1 - \gamma)} \vones ~.
\]
Considering the maximum of $\frac{\epsilon}{\alpha}$ and $\frac{\delta_1 M}{1 - \gamma}$ yields the result. 
\end{proof}

We now have all the pieces we need to prove our desired result.

\begin{theorem}[Interior Point Based DMDP Running Time]
\label{thm:dmdp_with_ip}
 Given a DMDP we can use \cite{lee2015efficient} to compute an $\epsilon$-approximate policy with high probability in time 
\[
\runtimeipnice
\]
\end{theorem}

\begin{proof}
The result follows from using Corollary~\ref{cor:l1_regress} to solve the DMDP $\ell_1$ problem.
By Lemma~\ref{lem:apx_policy} it suffices to have an $\frac{\epsilon (1 - \gamma)^2}{8 \S}$ approximate solution to the $\ell_1$ problem. Therefore, we choose $\alpha = \frac{\epsilon (1-\gamma)^4}{32 M^2 \S^2 }$ and $\epsilon' = \max\{\frac{\alpha \epsilon (1 - \gamma)^2}{8 \S}, \frac{\epsilon (1 - \gamma)^3}{16 M \S} \}$ and use Corollary~\ref{cor:l1_regress} to compute an $\epsilon'$-optimal solution to the DMDP $\ell_1$ problem. This can be used to compute an $\epsilon$-approximate policy through Lemma~\ref{lem:apx_policy}. The running time follows from Corollary~\ref{cor:l1_regress}, as $\mM = \ms^{-1} \ma \in \R^{(\states \times \actions) \times \states}$  and we solve to error $O(\frac{\epsilon' (1-\gamma)}{|S|M})$ as for $b$ in the corollary $\|b\|_1 = O(|S| M / (1 - \gamma))$.
\end{proof}

\end{document}